\documentclass[preprint,aps,floatfix,nofootinbib,superscriptaddress]{revtex4-1}
\pdfoutput=1

\usepackage{amsmath, amssymb, amsfonts, amsthm, latexsym, epsfig, mathrsfs, xcolor, bbm, slashed}

\usepackage[inline]{enumitem}

\usepackage{setspace}
\usepackage[marginal, multiple]{footmisc}

\usepackage[T1]{fontenc}
\usepackage[utf8]{inputenc}
\usepackage{lmodern}

\usepackage[colorlinks, allcolors=blue!70!black, linktocpage]{hyperref}

\numberwithin{equation}{section}

\usepackage{cleveref}

\usepackage{microtype}

\usepackage{floatrow}
\let\OLDtableofcontents\tableofcontents
\renewcommand\tableofcontents[1]{%
	{\baselineskip 0.5ex %
		\OLDtableofcontents{#1}}%
}

\let\OLDthebibliography\thebibliography
\renewcommand\thebibliography[1]{%
	\setstretch{1.079} % 1/2 goldenratio
	\OLDthebibliography{#1}%
	\small %
	\setlength{\itemsep}{0.2\baselineskip} % goldenratio for separating items %
}

\let\OLDfootnote\footnote
\renewcommand\footnote[1]{%
	\setlength{\footnotesep}{0.75\baselineskip}%
	{\footnotesize \OLDfootnote{#1}}%
}

\setlist[enumerate]{noitemsep, label=(\arabic*), ref=(\arabic*)}

\newlist{condlist}{enumerate}{2}
\setlist[condlist,1]{noitemsep, topsep=0pt, label=(\arabic*), ref=(\arabic*)}
\setlist[condlist,2]{noitemsep, label=(\alph*), ref=(\arabic{condlisti}.\alph*)}
\crefname{condlisti}{condition}{conditions}
\crefname{condlistii}{condition}{conditions}

\newlist{propertylist}{enumerate}{1}
\setlist[propertylist,1]{noitemsep, topsep=0pt, label=(\arabic*), ref=(\arabic*)}
\crefname{propertylisti}{Property}{Properties}

\renewcommand\thesection{\arabic{section}}
\renewcommand\thesubsection{\arabic{subsection}}

\makeatletter
\def\p@subsection{\thesection.}
\def\p@subsubsection{\thesection.\thesubsection.}
\makeatother 

\theoremstyle{plain}
\newtheorem{thm}{Theorem}
\newtheorem{lemma}{Lemma}[section]

\theoremstyle{definition}
\newtheorem{definition}{Definition}[section]

\theoremstyle{remark}
\newtheorem{remark}{Remark}[section]

\crefname{equation}{Eq.}{Eqs.}
\creflabelformat{equation}{#2#1#3}

\crefname{section}{Sec.}{Secs.}
\crefname{appendix}{Appendix}{Appendices}
\crefname{figure}{Fig.}{Figs.}
\crefname{table}{Table}{Tables}

\crefname{definition}{Def.}{Defs.}
\crefname{prop}{Prop.}{Props.}
\crefname{lemma}{Lemma}{Lemmas}
\crefname{corollary}{Cor.}{Cors.}
\crefname{thm}{Theorem}{Theorems}
\crefname{remark}{Remark}{Remarks}

\crefname{ass}{Assumptions}{Assumptions}
\crefname{property}{Properties}{Properties}

\newcommand{\be}{\begin{equation}}
\newcommand{\ee}{\end{equation}}

\newcommand{\lb}{\left}
\newcommand{\rb}{\right}

\newcommand{\lra}{\leftrightarrow}

\newcommand{\mc}{\mathcal}
\newcommand{\f}[2]{\frac{#1}{#2}}

\newcommand{\ms}{\mathscr}
\newcommand{\mf}{\mathfrak}
\newcommand{\bb}{\mathbb}

\newcommand{\eqsp}{\, ,\quad} %shrtct for space in eqns

\newcommand{\hr}{\begin{center}* * *\end{center}}

\newcommand{\Lie}{\pounds} %shrtct for Lie derivative
\newcommand{\lie}{\pounds} %shrtct for Lie derivative (alternate)
\newcommand{\defn}{\mathrel{\mathop:}=} %shrtct for definition operator

\newcommand{\union}{\cup} % union symbol
\newcommand{\inter}{\cap} % intersection

\let\oldsetminus\setminus
\renewcommand{\setminus}{\!\oldsetminus\!} 

\let\oldint\int
\renewcommand{\int}{\oldint\limits}

\let\oldlim\lim
\renewcommand{\lim}{\oldlim\limits}

\renewcommand{\bar}{\overline}

\newcommand{\scri}{\ms I}
\newcommand{\hyp}{\ms H}
\newcommand{\spi}{\mf{spi}}

\newcommand{\dd}[1]{\boldsymbol{#1}} %notation for dir-dep limits

\newcommand{\df}[1]{#1} %notation for diff. forms; don't use since not needed

\newcommand{\nfrac}[2]{{{}^#1\!\!/\!_#2}}

\newcommand{\half}{\nfrac{1}{2}}

\newcommand{\Omh}{\Omega^\half} % order notation

\newcommand{\pb}[1]{\underleftarrow{#1}} % shrct for pullback

\newcommand{\rsub}[1]{\scriptscriptstyle{\rm #1}}

\newcommand{\bs}{{\rsub{(BS)}}}

\newcommand{\dv}[1]{{}^\star\!{\dd{#1}}}

\begin{document}

\setstretch{1.2}

% Title Page

\title{Asymptotic symmetries and charges at spatial infinity in general relativity}

\author{Kartik Prabhu}
\email{kartikprabhu@cornell.edu}
\affiliation{Cornell Laboratory for Accelerator-based Sciences and Education (CLASSE),\\ Cornell University, Ithaca, NY 14853, USA}

\author{Ibrahim Shehzad}
\email{is354@cornell.edu}
\affiliation{Department of Physics, Cornell University, Ithaca, NY 14853, USA}

%%============================================================
\begin{abstract}
We analyze the asymptotic symmetries and their associated charges at spatial infinity in \(4\)-dimensional asymptotically-flat spacetimes. We use the covariant formalism of Ashtekar and Hansen where the asymptotic fields and symmetries live on the \(3\)-manifold of spatial directions at spatial infinity, represented by a timelike unit-hyperboloid (or de Sitter space). Using the covariant phase space formalism, we derive formulae for the charges corresponding to asymptotic supertranslations and Lorentz symmetries at spatial infinity. With the motivation of, eventually, proving that these charges match with those defined on null infinity --- as has been conjectured by Strominger --- we do not impose any restrictions on the choice of conformal factor in contrast to previous work on this problem. Since we work with a general conformal factor we expect that our charge expressions will be more suitable to prove the matching of the Lorentz charges at spatial infinity to those defined on null infinity, as has been recently shown for the supertranslation charges.
\end{abstract}

\maketitle
\tableofcontents

%%============================================================
\section{Introduction}
\label{sec:intro}

In general relativity, the asymptotic symmetries of asymptotically-flat spacetimes at both past and future null infinity are elements of the infinite-dimensional Bondi-Metzner-Sachs (BMS) group \cite{BBM, Sachs1} (see \cite{Ashtekar:2014zsa, AK} for recent reviews). It has been conjectured by Strominger \cite{Stro-CK-match} that the (a priori independent) BMS groups at past and future null infinity are related via an antipodal reflection near spatial infinity. This matching relation gives a \emph{global} ``diagonal'' asymptotic symmetry group for general relativity. If similar matching conditions relate the gravitational fields, then there exist infinitely many conservation laws in classical gravitational scattering between the incoming fluxes associated with the BMS group at past null infinity and the outgoing fluxes of the corresponding (antipodally identified) BMS group at future null infinity. These conservation laws are also related to soft graviton theorems \cite{He:2014laa,Strominger:2014pwa,Kapec:2015vwa,Avery:2015gxa,Campiglia:2015kxa,Campoleoni:2017mbt}, gravitational memory effects \cite{He:2014laa,Strominger:2014pwa,Pasterski:2015tva,HIW,Mao:2017wvx,Pate:2017fgt,Chatterjee:2017zeb} and the black hole information paradox \cite{Hawking:2016msc,Strominger:2017aeh,Hawking:2016sgy} (see \cite{Strominger:2017zoo} for a detailed review of recent developments and a complete list of references).

Such matching conditions on the asymptotic symmetries and fields have been shown in Maxwell theory on a background Minkowski spacetime \cite{CE} and in general asymptotically-flat spacetimes \cite{KP-EM-match}. In the gravitational case, the matching of the supertranslation symmetries and supermomentum charges has also be proven for linearized perturbations on a Minkowski background \cite{Tro} and in general asymptotically-flat spacetimes \cite{KP-GR-match}. For the translation symmetries these reduce to the much older result of \cite{Ash-Mag-Ash} which shows that the Bondi \(4\)-momentum on future and past null infinity matches the \(4\)-momentum at spatial infinity.

The main technique used in \cite{CE,KP-EM-match,Tro,KP-GR-match} to prove these matching conditions is to ``interpolate'' between the symmetries and charges at past and future null infinities using the field equations and the asymptotic symmetries and charges defined near spatial infinity. In a background Minkowski spacetime this analysis can be done using asymptotic Bondi-Sachs coordinates near each null infinity and asymptotic Beig-Schmidt coordinates near spatial infinity. Using the explicit transformations between these coordinate systems the matching conditions can be shown to hold for Maxwell fields and linearized gravity on Minkowski spacetime \cite{CE,Tro}. But in general asymptotically-flat spacetimes the transformations between the asymptotic coordinates is not known explicitly. In this case the covariant formulation of asymptotic-flatness given by Ashtekar and Hansen \cite{AH}, which treats both null and spatial infinities in a unified spacetime-covariant manner, has proven fruitful to analyze the matching of the symmetries and charges \cite{KP-EM-match,KP-GR-match}.

However, for the charges associated with the Lorentz symmetries such matching conditions between past and future null infinity have not yet been proven, except for the case of stationary spacetimes \cite{AS-ang-mom}. With an eye towards establishing these conjectured matching conditions for Lorentz symmetries and charges we revisit the formulation of the asymptotic symmetries and charges at spatial infinity.

The asymptotic behaviour at spatial infinity can be studied using many different (but related) formalisms. Since our primary motivation is to, ultimately, make contact with null infinity it will be more useful to use a spacetime covariant formalism without using any \((3+1)\) decomposition of the spacetime by spacelike hypersurfaces \cite{ADM, ADMG, RT-parity}. Such a \(4\)-dimensional formulation of asymptotic-flatness at spatial infinity can be be given using suitable asymptotic coordinates as formulated by Beig and Schmidt \cite{Beig-Schmidt}. The asymptotic symmetries and charges using the asymptotic expansion of the metric in these coordinates have been worked out in detail \cite{Beig-Schmidt, CDV-Lorentz, CD}. But as mentioned above, the relation between the Beig-Schmidt coordinates and the coordinates adapted to null infinity (like the Bondi-Sachs coordinates) is not known in general spacetimes. Thus, we will use the coordinate independent formalism of Ashtekar and Hansen \cite{AH,Ash-in-Held} (\cref{def:AH}) to investigate the symmetries and their associated charges at spatial infinity.\footnote{The relation between the Ashtekar-Hansen formalism and the Beig-Schmidt coordinates is summarized in \cref{sec:coord}.}

The asymptotic behaviour of the gravitational field for any asymptotically-flat spacetime is most conveniently described in a conformally-related unphysical spacetime, the Penrose conformal-completion. In the unphysical spacetime, null infinities \(\scri^\pm\) are smooth null boundaries while spatial infinity is a boundary point \(i^0\) which is the vertex of ``the light cone at infinity'' formed by \(\scri^\pm\). For Minkowski spacetime the unphysical spacetime is smooth (in fact, analytic) at \(i^0\). However, in more general spacetimes, the unphysical metric is not even once-differentiable at spatial infinity unless the ADM mass of the spacetime vanishes \cite{AH}, and the unphysical spacetime manifold \emph{does not} have a smooth differential structure at \(i^0\). Thus, in the Ashtekar-Hansen formalism, instead of working directly at the point \(i^0\) where sufficiently smooth structure is unavailable, one works on a ``blowup'' --- the space of spatial directions at \(i^0\) --- given by a timelike-unit-hyperboloid \(\hyp\) in the tangent space at \(i^0\). Suitably conformally rescaled fields, whose limits to \(i^0\) depend on the direction of approach, induce \emph{smooth} fields on \(\hyp\) and we can study these smooth limiting fields using standard differential calculus on \(\hyp\). For instance, in Maxwell theory the rescaled field tensor \(\Omega F_{ab}\) and in general relativity the rescaled (unphysical) Weyl tensor \(\Omh C_{abcd}\) (where \(\Omega\) is the conformal factor used in the Penrose conformal completion) admit regular direction-dependent limits to \(i^0\), and these fields induce smooth tensor fields on \(\hyp\). Similarly, the Maxwell gauge transformations and vector fields in the physical spacetime (suitably rescaled) admit regular direction-dependent limits which generate the asymptotic symmetries at \(i^0\) (see \cref{sec:spi-symm}).

The asymptotic symmetries in general relativity at spatial infinity have also been studied in detail in the Ashtekar-Hansen formalism \cite{AH,Ash-in-Held}. However in deriving the charges associated with these symmetries Ashtekar and Hansen reduced the asymptotic symmetry algebra from the infinite-dimensional \(\spi\) algebra to the Poincar\'e algebra consisting only of translations and Lorentz transformations. This reduction was accomplished by demanding that the ``leading order'' magnetic part of the Weyl tensor, given by a tensor \(\dd B_{ab}\) on \(\hyp\) (see \cref{eq:EB-defn}), vanish and additionally choosing the conformal factor near \(i^0\) so that the tensor potential \(\dd K_{ab}\) for \(\dd B_{ab}\) also vanishes (see \cref{rem:GR-gauge-choice}). This restriction was also imposed in \cite{MMMV-Lorentz,CDV-Lorentz}. In the work of Comp\`ere and Dehouck in \cite{CD}, the condition \(\dd B_{ab} = 0\) was not imposed however, they also specialized to a conformal factor where the trace \(\dd h^{ab} \dd K_{ab}\) (where $\dd{h}^{ab}$ denotes the inverse of the metric on $\hyp$) was set to vanish. As we will show below (see \cref{sec:conf-charges}) the charges of the Lorentz symmetries at spatial infinity are not conformally-invariant but shift by the charge of a supertranslation. This is entirely analogous to the supertranslation ambiguities in the Lorentz charges at null infinity. Thus, when matching the Lorentz charges at spatial infinity to those at past and future null infinity, one would need to perform this matching in the ``same'' choice of conformal factor in all three regions. A priori, it is not clear what the special choices of conformal factor chosen in the above mentioned analyses imply at null infinity. Thus, we will not impose any such restrictions on the conformal factor and not impose any conditions on \(\dd K_{ab}\) (apart from its equations of motion arising from the Einstein equation) in our analysis. As we will show, one peculiar consequence of keeping a completely unrestricted conformal factor will be that our charges will not be exactly conserved but will have a non-vanishing flux through regions of \(\hyp\) (except for pure translations). Thus, these charges are not associated with the point \(i^0\) at spatial infinity, but with cross-sections of the ``blowup'' \(\hyp\). This is not a serious drawback; as shown in \cite{KP-EM-match,KP-GR-match} for matching the symmetries and charges at null infinity, one only requires that the \emph{total} flux of the charges through all of \(\hyp\) vanish but there can be a non-vanishing flux through local regions of \(\hyp\). Thus, our main goal in this work is to analyze the symmetries and charges in general relativity without imposing any restrictions on the choice of conformal factor near spatial infinity.

\hr

In our analysis of the asymptotic charges we will use the \emph{covariant phase space} formalism described below. Since the relevant quantities in the covariant phase space are defined in terms of the physical metric and their perturbations, we first analyze the conditions on the corresponding unphysical quantities so that they preserve the asymptotic-flatness conditions and the universal structure at \(i^0\) (\cref{sec:pert}). To derive the asymptotic symmetry algebra we then consider a physical metric perturbation \(\Lie_\xi \hat g_{ab}\) generated by an infinitesimal diffeomorphism and demand that it preserve the asymptotic conditions in the unphysical spacetime in the limit to \(i^0\). This will provide us with the following description of the asymptotic symmetries at \(i^0\) (\cref{sec:spi-symm}). The asymptotic symmetry algebra \(\spi\) is parametrized by a pair \((\dd f, \dd X^a)\) where \(\dd f\) is any smooth function and \(\dd X^a\) is a Killing field on \(\hyp\). The function \(\dd f\) parametrizes the supertranslations and \(\dd X^a\) parametrize the Lorentz symmetries. The \(\spi\) algebra is then a semi-direct sum of the Lorentz algebra with the infinite-dimensional abelian subalgebra of supertranslations. Note that this is the same as the asymptotic symmetries derived in \cite{AH,Ash-in-Held}. The only difference in our analysis is that we obtain the symmetries by analyzing the conditions on diffeomorphisms in the physical spacetime instead of using the unphysical spacetime directly as in \cite{AH,Ash-in-Held}.

To obtain the charges associated with these symmetries, the primary quantity of interest is the \emph{symplectic current} derived from the Lagrangian of a theory (see, \cite{LW,WZ} for details). The symplectic current \(\df\omega(\hat g; \delta_1 \hat g, \delta_2\hat g)\), is a local and covariant \(3\)-form which is an antisymmetric bilinear in two metric perturbations, $\delta \hat g$ on the physical spacetime. It can be shown that when the second perturbation \(\delta \hat g_{ab} = \Lie_\xi \hat g_{ab}\) is the perturbation corresponding to an infinitesimal diffeomorphism generated by a vector field \(\xi^a\) we have
\be \label{eq:sympcurrent-charge}
    \omega(\hat{g}; \delta \hat{g}, \lie_{\xi} \hat{g}) = d[\delta Q_{\xi} - \xi \cdot \theta(\delta \hat g)]\,,
\ee
where we have assumed that \(\hat g_{ab}\) satisfies the equations of motion and \(\delta \hat g_{ab}\) satisfies the linearized equations of motion. The \(2\)-form \(Q_\xi\) is the \emph{Noether charge} associated with the vector field \(\xi^a\) and the \(3\)-form \(\theta(\delta \hat g)\) is the \emph{symplectic potential} \cite{LW,WZ}. If we integrate \cref{eq:sympcurrent-charge} over a \(3\)-dimensional surface \(\Sigma\) with boundary \(\partial\Sigma\) we get
\begin{equation} \label{eq:pertcharge}
    \int_{\Sigma}\omega[\hat{g};\delta \hat{g}, \lie_{\xi}\hat{g}] = \int_{\partial \Sigma} \delta Q_{\xi} - \xi \cdot \theta(\delta \hat g) \,.
\end{equation}
To define the asymptotic charges at spatial infinity, we would like to evaluate \cref{eq:pertcharge} when the surface \(\Sigma\) extends to a suitably regular \(3\)-surface at \(i^0\) in the unphysical spacetime. Given the low amount of differentiability at \(i^0\) the appropriate condition is that \(\Sigma\) extends to a \(C^{>1}\) surface at \(i^0\). The limit of the boundary \(\partial\Sigma\) to \(i^0\) corresponds to a \(2\)-sphere cross-section \(S\) of the unit-hyperboloid \(\hyp\) in the Ashtekar-Hansen formalism. Then, the limiting integral on the right-hand-side of \cref{eq:pertcharge} (with the asymptotic conditions imposed on the metric perturbations as well as the symmetries) will define a perturbed charge on \(S\) associated with the asymptotic symmetry generated by \(\xi^a\). However, even though the explicit expressions for the integrand on the right-hand-side of \cref{eq:pertcharge} are well-known (see for instance \cite{WZ}), computing this limiting integral is difficult. So we will use an alternative strategy described next.

We will show that with the appropriate asymptotic-flatness conditions at \(i^0\), the symplectic current \(3\)-form \(\omega \equiv \omega_{abc}\) is such that \(\Omega^{\nfrac{3}{2}} \omega_{abc}\) has a direction-dependent limit to \(i^0\). The pullback of this limit to \(\hyp\), which we denote by \(\pb{\dd\omega}\), defines a symplectic current on \(\hyp\). We show that when one of the perturbations in this symplectic current is generated by an asymptotic \(\spi\) symmetry \((\dd f, \dd X^a)\), we have
\be
    \pb{\dd\omega}(g; \delta g, \delta_{(\dd f, \dd X)} g) = - \dd\varepsilon_3 \dd D^a \dd Q_a (g; \delta g, (\dd f, \dd X))\,,
\ee
where \(\dd\varepsilon_3\) and \(\dd D\) are the volume element and covariant derivative on \(\hyp\). The covector \(\dd Q_a (g; \delta g, (\dd f, \dd X))\) is a local and covariant functional of the background fields corresponding to the asymptotic (unphysical) metric \(g_{ab}\), and linear in the asymptotic (unphysical) metric perturbations \(\delta g_{ab}\) and the asymptotic symmetry parametrized by \((\dd f, \dd X^a)\). Thus, we can write the symplectic current, with one perturbation generated by an asymptotic symmetry, as a total derivative on \(\hyp\). Then, in analogy to \cref{eq:pertcharge}, we define the perturbed charge on a cross-section \(S\) of \(\hyp\) by the integral
\be\label{eq:pertcharge-hyp}
    \int_S \dd\varepsilon_2 \dd u^a \dd Q_a (g; \delta g, (\dd f, \dd X)\,,
\ee
where \(\dd\varepsilon_2\) is the area element and \(\dd u^a\) is a unit-timelike normal to the cross-section \(S\) within \(\hyp\). We then show that when the asymptotic symmetry is a supertranslation \(\dd f\), the quantity \(\dd Q_a (g; \delta g, \dd f)\) is integrable, i.e, it can be written as the \(\delta\) of some covector which is itself a local and covariant functional of the asymptotic fields and supertranslation symmetries. Then ``integrating'' \cref{eq:pertcharge-hyp} in the space of asymptotic fields, we can define a charge associated with the supertranslations on any cross-section \(S\) of \(\hyp\) (see \cref{sec:st}). When the asymptotic symmetry is a Lorentz symmetry parameterized by a Killing vector field \(\dd X^a\) on \(\hyp\), \cref{eq:pertcharge-hyp} \emph{cannot} be written as the \(\delta\) of some quantity (unless we restrict to the choice of conformal factor where \(\dd h^{ab} \dd K_{ab} = 0\) as described above). In this case, we will adapt the prescription by Wald and Zoupas \cite{WZ} to define an integrable charge for Lorentz symmetries (\cref{sec:lorentz-charge}). Then the change of these charges over a region \(\Delta\hyp\) bounded by two cross-sections provides a flux formula for these charges. In general, these fluxes will be non-vanishing (except for translation symmetries) unless we again restrict to the conformal factor where \(\dd h^{ab} \dd K_{ab} = 0\). However, as mentioned above, from the point of view of matching these charges to those on null infinity, the special conformal choices might not be convenient and it is not necessary to have exactly conserved charges on \(\hyp\). Thus, we will not restrict the conformal factor in any way and work with charges which can have non-trivial fluxes through some region of \(\hyp\).

\hr

The rest of this paper is organized as follows. In \cref{sec:AH} we recall the definition of asymptotic-flatness at spatial infinity in terms of an Ashtekar-Hansen structure. To illustrate our approach outlined above we first study the simpler case of Maxwell fields at spatial infinity, and derive the associated symmetries and charges in \cref{sec:Maxwell}. In \cref{sec:GR-eom} we then consider the asymptotic gravitational fields and Einstein equations at spatial infinity. We also describe the universal structure, that is the structure that is common to all spacetimes which are asymptotically-flat at \(i^0\), in \cref{sec:univ-str}. In \cref{sec:pert} we analyze the conditions on metric perturbations which preserve asymptotic flatness and obtain the limiting form of the symplectic current of general relativity on the space of directions \(\hyp\). In \cref{sec:spi-symm}, using the analysis of the preceding section, we derive the asymptotic symmetry algebra (the \(\spi\) algebra) by considering infinitesimal metric perturbations generated by diffeomorphisms which preserve the asymptotic flatness conditions. In \cref{sec:spi-charges} we derive the charges and fluxes corresponding to these \(\spi\) symmetries. We end with a summary and describe possible future directions in \cref{sec:disc}.

We collect some useful results and asides in the appendices. In \cref{sec:coord} we construct a useful coordinate system near \(i^0\) using the asymptotic flatness conditions on the unphysical metric and relate it to the Beig-Schmidt coordinates in the physical spacetime. \cref{sec:useful-H} collects useful results on the unit-hyperboloid \(\hyp\) on Killing vector fields, symmetric tensor fields and a theorem by Wald showing that (with suitable conditions) closed differential forms are exact. Computations detailing the change in the Lorentz charge under conformal transformations are presented in \cref{sec:conf-lorentz-charge}. In \cref{sec:amb} we show that our charges are unambiguously defined by the the symplectic current of vacuum general relativity. In \cref{sec:new-beta} we generalize the Lorentz charges derived in \cref{sec:lorentz-charge} to include spacetimes where the ``leading order'' magnetic part of the Weyl tensor \(\dd B_{ab}\) is allowed to be non-vanishing.

%%%-----------------------------------------------------------
%\subsection*{Notation and conventions}

\hr

We use an abstract index notation with indices \(a,b,c,\ldots\) for tensor fields. Quantities defined on the physical spacetime will be denoted by a ``hat'', while the ones on the conformally-completed unphysical spacetime are without the ``hat'' e.g. \(\hat g_{ab}\) is the physical metric while \(g_{ab}\) is the unphysical metric on the conformal-completion. We denote the spatial directions at \(i^0\) by \(\vec\eta\). Regular direction-dependent limits of tensor fields, which we will denote to be $C^{>-1}$, will be represented by a boldface symbol e.g. \(\dd C_{abcd}(\vec\eta)\) is the limit of the (rescaled) unphysical Weyl tensor along spatial directions at \(i^0\). The rest of our conventions follow those of Wald \cite{Wald-book}.

%%==============================================================
\section{Asymptotic-flatness at spatial infinity: Ashtekar-Hansen structure}
\label{sec:AH}

We define spacetimes which are asymptotically-flat at null and spatial infinity using an Ashtekar-Hansen structure \cite{AH, Ash-in-Held}. We use the following the notation for causal structures from \cite{Hawking-Ellis}: \(J(i^0)\) is the causal future of a point \(i^0\) in \(M\), \(\bar J(i^0)\) is its closure, \(\dot J(i^0) \) is its boundary and \(\scri \defn \dot J(i^0) - i^0\). We also use the definition and notation for direction-dependent tensors from \cite{Herb-dd}, see also Appendix B of \cite{KP-GR-match}.

\begin{definition}[Ashtekar-Hansen structure \cite{Ash-in-Held}]\label{def:AH}
	A \emph{physical} spacetime \((\hat M, \hat g_{ab})\) has an \emph{Ashtekar-Hansen structure} if there exists another \emph{unphysical} spacetime \((M, g_{ab})\), such that
	\begin{condlist}
		\item \(M\) is \(C^\infty\) everywhere except at a point \(i^0\) where it is \(C^{>1}\),
		\item the metric \(g_{ab}\) is \(C^\infty\) on \(M-i^0\), and \(C^0\) at \(i^0\) and \(C^{>0}\) along spatial directions at \(i^0\),
		\item there is an embedding of \(\hat M\) into \(M\) such that \(\bar J(i^0) = M - \hat M\),
		\item there exists a function \(\Omega\) on \(M\), which is \(C^\infty\) on \(M-i^0\) and \(C^2\) at \(i^0\) so that \(g_{ab} = \Omega^2 \hat g_{ab}\) on \(\hat M\) and
		\begin{condlist}
			\item \(\Omega = 0\) on \(\dot J(i^0)\),
			\item \(\nabla_a \Omega \neq 0\) on \(\scri\),
			\item at \(i^0\), \(\nabla_a \Omega = 0\), \(\nabla_a \nabla_b \Omega = 2 g_{ab}\). \label{cond:Omega-at-i0}
		\end{condlist}
		\item There exists a neighbourhood \(N\) of \(\dot J(i^0)\) such that \((N, g_{ab})\) is  strongly causal and time orientable, and in \(N \inter \hat M\) the physical metric \(\hat g_{ab}\) satisfies the vacuum Einstein equation \(\hat R_{ab} = 0\),
		\item The space of integral curves of \(n^a = g^{ab}\nabla_b \Omega\) on \(\dot J(i^0)\) is diffeomorphic to the space of null directions at \(i^0\), \label{cond:int-curves}
		\item The vector field \(\varpi^{-1} n^a\) is complete on \(\scri\) for any smooth function \(\varpi\) on \(M - i^0\) such that \(\varpi > 0\) on \(\hat M \union \scri\) and \(\nabla_a(\varpi^4 n^a) = 0\) on \(\scri\). \label{cond:complete}
	\end{condlist}
\end{definition}

The physical role of the conditions in \cref{def:AH} is to ensure that the point \(i^0\) is spacelike related to all points in the physical spacetime \(\hat M\), and represents \emph{spatial infinity}, and that null infinity \(\scri \defn \dot J(i^0) - i^0\) has the usual structure.
Note that the metric \(g_{ab}\) is only \(C^{>0}\) at \(i^0\) along spatial directions, that is, the metric is continuous but the metric connection is allowed to have limits which depend on the direction of approach to \(i^0\). This low differentiability structure is essential to allow spacetimes with non-vanishing ADM mass \cite{AH, Ash-in-Held}. In the following we will only consider the behaviour of the spacetime approaching \(i^0\) along spatial directions, and we will not need the conditions corresponding to null infinity.

\hr

For spacetimes satisfying \cref{def:AH} we have the following limiting structures at \(i^0\) when approached along spatial directions.

Along spatial directions \(\eta_a \defn \nabla_a \Omh\) is \(C^{>-1}\) at \(i^0\) and 
\be\label{eq:eta-defn}
\dd\eta^a \defn \lim_{\to i^0} \nabla^a \Omh\,,
\ee
determines a \(C^{>-1}\) spatial unit vector field at \(i^0\) representing the spatial directions \(\vec\eta\) at \(i^0\). The space of directions \(\vec\eta\) in \(Ti^0\) is a unit-hyperboloid \(\hyp\).

If \(T^{a \ldots}{}_{b \ldots}\) is a \(C^{>-1}\) tensor field at \(i^0\) in spatial directions then, \(\lim_{\to i^0} T^{a \ldots}{}_{b \ldots} = \dd T^{a \ldots}{}_{b \ldots}(\vec\eta)\) is a smooth tensor field on \(\hyp\). Further, the derivatives of \(\dd T^{a \ldots}{}_{b \ldots}(\vec\eta)\) to all orders with respect to the direction \(\vec\eta\) satisfy\footnote{The factors of \(\Omh\) on the right-hand-side of \cref{eq:dd-der-spatial} convert between \(\nabla_a\) and the derivatives with respect to the directions; see \cite{Ash-in-Held,Geroch-asymp}.}
\be\label{eq:dd-der-spatial}
    \dd \partial_c \cdots \dd \partial_d \dd T^{a \ldots}{}_{b \ldots}(\vec\eta) = \lim_{\to i^0} \Omh \nabla_c \cdots \Omh \nabla_d T^{a \ldots}{}_{b \ldots}\,,
\ee
where \(\dd \partial_a\) is the derivative with respect to the directions \(\vec \eta\) defined by 
\be\label{eq:dd-derivative-spatial}\begin{split}
	\dd v^c \dd \partial_c \dd T^{a \ldots}{}_{b \ldots}(\vec\eta) & \defn \lim_{\epsilon \to 0} \frac{1}{\epsilon} \big[ \dd T^{a \ldots}{}_{b \ldots}(\vec\eta + \epsilon \vec v) - \dd T^{a \ldots}{}_{b \ldots}(\vec\eta) \big] \quad \text{for all } \dd v^a \in T\hyp \,,\\
	\dd \eta^c \dd \partial_c \dd T^{a \ldots}{}_{b \ldots}(\vec\eta) & \defn 0\,.
\end{split}\ee

The metric \(\dd h_{ab}\) induced on \(\hyp\) by the universal metric \(\dd g_{ab}\) at \(i^0\), satisfies
\be\label{eq:d-eta-h}
    \dd h_{ab} \defn \dd g_{ab} - \dd \eta_a \dd \eta_b = \dd \partial_a \dd \eta_b\,.
\ee
Further, if \(\dd T^{a \ldots}{}_{b \ldots}(\vec\eta)\) is orthogonal to \(\dd\eta^a\) in all its indices then it defines a tensor field \(\dd T^{a \ldots}{}_{b \ldots}\) intrinsic to \(\hyp\). In this case, it follows from \cref{eq:d-eta-h} and \(\dd\partial_c \dd g_{ab} = 0\) (since \(\dd g_{ab}\) is direction-independent at \(i^0\)) that projecting \emph{all} the indices in \cref{eq:dd-der-spatial} using \(\dd h_{ab}\) defines a derivative operator \(\dd D_a\) intrinsic to \(\hyp\) which is also the covariant derivative operator associated with \(\dd h_{ab}\). We also define
\be\label{eq:volume-hyp}
\dd\varepsilon_{abc} \defn - \dd\eta^d \dd\varepsilon_{dabc} \eqsp \dd\varepsilon_{ab} \defn \dd u^c \dd\varepsilon_{cab}\,,
\ee
where \(\dd\varepsilon_{abcd}\) is volume element at \(i^0\) corresponding to the metric \(\dd g_{ab}\), \(\dd\varepsilon_{abc}\) is the induced volume element on \(\hyp\), and \(\dd\varepsilon_{ab}\) is the induced area element on some cross-section \(S\) of \(\hyp\) with a future-pointing timelike normal \(\dd u^a\) such that \(\dd h_{ab} \dd u^a \dd u^b = -1\). 

\begin{remark}[Conformal freedom]\label{rem:conf}
    It follows from the conditions in \cref{def:AH} that the allowed conformal freedom \(\Omega \mapsto \omega \Omega\) is such that \(\omega > 0\) is smooth in \(M - i^0\), is \(C^{>0}\) at \(i^0\) and \(\omega\vert_{i^0} = 1\). From these conditions it follows that 
\be \label{eq:conf-tr-defn}
 \omega = 1 + \Omh \alpha \,,
 \ee
where \(\alpha\) is \(C^{>-1}\) at \(i^0\). Let \(\dd\alpha(\vec\eta) \defn \lim_{\to i^0} \alpha\), then from \cref{eq:conf-tr-defn} we also get
\be
    \lim_{\to i^0} \nabla_a \omega = \dd\alpha \dd\eta_a + \dd D_a \dd\alpha\,.
\ee
Note in particular, that the unphysical metric \(\dd g_{ab}\) at \(i^0\) is invariant under conformal transformations. While
\be
    \eta^a \mapsto \omega^{-2}[\omega^\half \eta^a + \tfrac{1}{2} \omega^{-\half} \Omh \nabla^a \omega ] \implies \dd\eta^a \mapsto \dd\eta^a\,.
\ee
Thus, unit spatial directions \(\vec\eta\), the space of directions \(\hyp\), and the induced metric on it \(\dd h_{ab}\) are also invariant.
\end{remark}

%%===========================================================
\section{Maxwell fields: symmetries and charges at \(i^0\)}
\label{sec:Maxwell}

To illustrate our general strategy, we first consider the simpler case of Maxwell fields on any fixed background spacetime satisfying \cref{def:AH}.

In the physical spacetime \(\hat M\), let \(\hat F_{ab}\) be the Maxwell field tensor satisfying the Maxwell equations 
\be\label{eq:Max-phys}
    \hat g^{ac} \hat g^{bd}\hat\nabla_b \hat F_{dc} = 0 \eqsp \hat\nabla_{[a} \hat F_{bc]} = 0\,.
\ee
In the unphysical spacetime \(M\) with \(F_{ab} \defn \hat F_{ab}\) we have 
\be\label{eq:maxwell}
    \nabla_b F^{ba} = 0 \eqsp \nabla_{[a} F_{bc]} = 0\,.
\ee

The Maxwell tensor $F_{ab}$ is smooth everywhere in the unphysical spacetime except at $i^{0}$. Analyzing the behaviour of \(F_{ab}\) in the simple case of a static point charge in Minkowski spacetime, it can be seen that \(F_{ab}\) diverges in the limit to $i^{0}$, but $\Omega F_{ab}$ admits a direction-dependent limit.\footnote{Note that this diverging behaviour of \(F_{ab}\) refers to the tensor in the unphysical spacetime with the chosen \(C^{>1}\) differential structure at \(i^0\). In an asymptotically Cartesian coordinate system of the physical spacetime, this behaviour reproduces the standard \(1/r^2\) falloff for \(F_{ab}\) and \(\dd F_{ab}(\vec\eta)\) is the ``leading order'' piece at \(O(1/r^2)\).}  Hence we assume as our asymptotic condition that
\be \label{eq:dd-F}
 \lim_{\to i^{0}}\Omega	F_{ab} =  \dd{F}_{a b} (\vec\eta) \text{ is } C^{>-1}\,.
\ee
The direction-dependent limit of the Maxwell tensor, $\dd{F}_{ab}$, induces smooth tensor fields on $\hyp$. These are given by the ``electric'' and ``magnetic'' parts of the Maxwell tensor defined by
\be\label{eq:EB-F-defn}
    \dd{E}_{a}(\vec\eta) =\dd{F}_{ab}(\vec\eta) \dd{\eta}^{b} \eqsp \dd{B}_{a}(\vec\eta) = * \dd{F}_{ab}(\vec\eta) \dd{\eta}^{b}\,.
\ee 	
where \(* \dd{F}_{ab}(\vec\eta) \defn \tfrac{1}{2} \dd\varepsilon_{ab}{}^{cd} \dd F_{cd}(\vec\eta) \) is the Hodge dual with respect to the unphysical volume element \(\dd\varepsilon_{abcd}\) at \(i^0\). The electric and magnetic fields are orthogonal to \(\dd\eta^a\) and thus induce intrinsic fields \(\dd E_a\) and \(\dd B_a\) on \(\hyp\). Note that $\dd{F}_{ab}$ can be reconstructed from $\dd{E}_{a}$ and $\dd{B}_{a}$ using
\be
    \dd{F}_{ab} = 2 \dd{E}_{[a} \dd{\eta}_{b]} + \dd{\varepsilon}_{abcd} \dd{\eta}^{c} \dd{B}^{d}\,.
\ee
The asymptotic Maxwell equations are obtained by multiplying \cref{eq:maxwell} by \(\Omega^{\nfrac{3}{2}}\) and taking the limit to \(i^0\) in spatial directions (see \cite{AH} for details)
\be\label{eq:Max-eqn-asymp}\begin{aligned}
\dd D^{a}\dd{E}_{a} &=0 \eqsp \dd D_{[a} \dd{E}_{b]} =0 \,, \\
\dd D^{a} \dd{B}_{a} &=0 \eqsp \dd D_{[a}\dd{B}_{b]} =0\,.
\end{aligned}\ee\\

To use the symplectic formalism for Maxwell theory, we will need to introduce the vector potential as the basic dynamical field. Let \(\hat A_a\) be a vector potential for \(\hat F_{ab}\) so that \(\hat F_{ab} = 2 \hat\nabla_{[a} \hat A_{b]}\) in the physical spacetime. Then, \(A_a \defn \hat A_a\) is a vector potential for \(F_{ab}\) in the unphysical spacetime. We further assume that the vector potential \(A_a\) for \(F_{ab}\) is chosen so that \(\Omh A_a\) is \(C^{>-1}\) at \(i^0\). Then define the asymptotic potentials
\be\label{eq:EM-potentials}
    \dd V(\vec\eta) \defn \dd\eta^a \lim_{\to i^0} \Omh A_a \eqsp \dd A_{a}(\vec\eta) \defn \dd h_a{}^b \lim_{\to i^0} \Omh A_b\,.
\ee
Then the corresponding smooth fields \(\dd V\) and \(\dd A_a\) induced on \(\hyp\) act as potentials for the electric and magnetic field through
\be\label{eq:Max-EB-potentials}
    \dd E_a = \dd D_a \dd V \eqsp \dd B_a = \tfrac{1}{2} \dd\varepsilon_a{}^{bc} \dd D_b \dd A_c\,.
\ee
Even though we do not need this form, for completeness, we note that the Maxwell equations on \(\hyp\) (\cref{eq:Max-eqn-asymp}) can be written in terms of the potentials \(\dd V\) and \(\dd A_a\) as
\be
    \dd D^2 \dd V = 0 \eqsp \dd D^2 \dd A_a = \dd D_a \dd D^b \dd A_b + 2 \dd A_a\,.
\ee\\

Now consider a gauge transformation of the vector potential 
\be\label{eq:EM-gauge}
    A_a \mapsto A_a + \nabla_a \lambda\,,
\ee
where \(\lambda\) is \(C^{>-1}\) at \(i^0\). Then with \(\dd\lambda(\vec\eta) \defn \lim_{\to i^0} \lambda\), the gauge transformations of the asymptotic potentials (\cref{eq:EM-potentials}) on \(\hyp\) is given by
\be\label{eq:Max-symm}
    \dd V \mapsto \dd V \eqsp \dd A_a \mapsto \dd A_a + \dd D_a \dd\lambda\,.
\ee
Thus, the asymptotic symmetries of Maxwell fields at \(i^0\) are given by the functions \(\dd\lambda\) on \(\hyp\).

\begin{remark}[Special choices of gauge]\label{rem:EM-gauge-choice}
The gauge freedom in the Maxwell vector potential can be used to impose further restrictions on the potential \(\dd A_a\) on \(\hyp\). We illustrate the following two gauge conditions which will have analogues in the gravitational case (see \cref{rem:GR-gauge-choice}).
\begin{enumerate}
    \item Consider the Lorenz gauge condition \(\hat g^{ab} \hat\nabla_a \hat A_b = 0\) on the physical vector potential \(\hat A_a\) in the physical spacetime as used in \cite{CE,HT-em}. Multiplying this condition by \(\Omega^{-1}\) and taking the limit to \(i^0\), using \cref{eq:EM-potentials} we get the asymptotic gauge condition
\be\label{eq:Lorenz-gauge-H}
    \dd D^a \dd A_a = 2 \dd V\,.
\ee
Alternatively, from \cref{eq:Max-symm} we see that
\be
    \dd D^a \dd A_a \mapsto \dd D^a \dd A_a + \dd D^2 \dd\lambda\,.
\ee
    By solving a linear hyperbolic equation for \(\dd\lambda\) we can choose a new gauge in which
    \be
        \dd D^a \dd A_a = 0\,.
    \ee
    Both these gauge conditions reduce the allowed asymptotic symmetries to
    \be
        \dd D^2 \dd\lambda = 0\,.
    \ee
    \item If we impose the restriction \(\dd B_a = 0\) then \(\dd D_{[a}\dd A_{b]} = 0\) and thus there exists a function \(\dd A\) so that \(\dd A_a = \dd D_a \dd A\).\footnote{This follows from the fact that every \(1\)-loop in \(\hyp\) is contractible to a point and hence the first de Rahm cohomology group of \(\hyp\) is trivial.} Then using the transformation \cref{eq:Max-symm} we can set \(\dd A_a = 0\). The remaining asymptotic symmetries are just the Coulomb symmetries \(\dd \lambda = \text{constant}\). This is analogous to the condition used by Ashtekar and Hansen in the gravitational case to reduce the asymptotic symmetries to the Poincar\'e algebra \cite{AH}.
\end{enumerate}
    In what follows we will not need to impose any gauge condition on the potential \(\dd A_a\) and our analysis will be completely gauge invariant.
\end{remark}

\begin{remark}[Logarithmic gauge transformations]
\label{rem:log-gauge}
Note that above we only considered gauge transformations \cref{eq:EM-gauge} where the gauge parameter \(\lambda\) was \(C^{>-1}\) at \(i^0\). However, there is an additional ambiguity in the choice of gauge given by the \emph{logarithmic gauge transformations} of the form
\be
    A_a \mapsto A_a + \nabla_a (\ln\Omh \Lambda)\,,
\ee
where \(\Lambda\) is \(C^{>0}\) at \(i^0\). Under this gauge transformation \(\Omh A_a\) is still \(C^{>-1}\) at \(i^0\), and from \cref{eq:EM-potentials} we have the transformations
\be
    \dd V \mapsto \dd V + \dd\Lambda \eqsp \dd A_a \mapsto \dd A_a \,,
\ee
where \(\dd\Lambda \defn \lim_{\to i^0} \Lambda\) which is direction-independent at \(i^0\) and induces a constant function on \(\hyp\). From \cref{eq:Max-EB-potentials} we see that the fields \(\dd E_a\) and \(\dd B_a\) are invariant under this transformation. Since our charges and fluxes, derived below, will be expressed in terms of \(\dd E_a\) we will not need to fix this logarithmic gauge ambiguity in the potentials for electromagnetism. However, there is an analogous logarithmic translation ambiguity in the gravitational case which we will need to fix (see \cref{rem:log-trans}). Thus we now illustrate how this logarithmic gauge ambiguity can be fixed even in electromagnetism.

Since the metric \(\dd g_{ab}\) in the tangent space \(Ti^0\) is universal and isometric to the Minkowski metric it is invariant under the reflection of the spatial directions \(\vec\eta \mapsto - \vec\eta\). This gives rise to a reflection isometry of the metric \(\dd h_{ab}\) on the space of directions \(\hyp\). It was shown in \cite{KP-EM-match} that the Maxwell fields on \(\hyp\) which ``match'' on to asymptotically-flat Maxwell fields on null infinity are the ones where the electric field \(\dd E_a\) is reflection-odd i.e.
\be
    \dd E_a (\vec\eta) = - \dd E_a(-\vec\eta)\,.
\ee
Further, since the logarithmic gauge parameter \(\dd\Lambda\) is \emph{direction-independent} we have that, \(\dd\Lambda\) is reflection-even
\be
    \dd \Lambda(\vec\eta) = \dd\Lambda(-\vec\eta)\,.
\ee
Using a reflection-odd \(\dd E_a\) in \cref{eq:Max-EB-potentials} we see that using a logarithmic gauge transformation we can demand that the potential \(\dd V\) is also reflection-odd, so that 
\be
    \dd V (\vec\eta) = - \dd V(-\vec\eta)\,.
\ee
This fixes the logarithmic gauge ambiguity in the potentials.
\end{remark}

\hr

Let us now analyze the charges and fluxes for this theory. To do this, we start by studying the symplectic current. In vacuum electromagnetism, this is given by:
\begin{equation}\label{eq:Maxwell-symp}
    \omega_{abc}(\delta_{1} A, \delta_{2}A) = \hat\varepsilon_{abcd} \lb( \delta_{1} \hat{F}^{de} \delta_{2} \hat A_{e} - \delta_{2}\hat {F}^{de} \delta_{1} \hat A_{e} \rb)\,, 
\end{equation}
where the indices on \(\delta \hat F_{ab}\) have been raised with the physical metric \(\hat g^{ab}\). In terms of quantities in the unphysical spacetime we have
\be\label{eq:Maxwell-symp-unphys}
    \omega_{abc}(\delta_{1} A, \delta_{2} A) = \varepsilon_{abcd} \lb( \delta_{1} {F}^{de} \delta_{2} A_{e} - \delta_{2} {F}^{de} \delta_{1} A_{e} \rb) \,,
\ee
where we have used $\hat{\varepsilon}_{abcd} = \Omega^{-4} \varepsilon_{abcd}\,,$ and \(\hat g^{ab} = \Omega^2 g^{ab}\). 

To obtain the limit to \(i^0\) we rewrite this in terms of direction-dependent quantities from  \cref{eq:dd-F,eq:EM-potentials}. We see that \(\Omega^{\nfrac{3}{2}}\omega_{abc}\) is \(C^{>-1}\) at \(i^0\). The pullback of this direction-dependent limit to \(\hyp\) is then given by
\be
    \pb{\dd\omega}(\delta_{1} A, \delta_{2} A) = - \dd\varepsilon_3 \lb( \delta_1 \dd E^a \delta_2 \dd A_a - \delta_2 \dd E^a \delta_1 \dd A_a \rb)\,,
\ee
where $\dd{\varepsilon}_{3} = \dd{\varepsilon}_{abc}$ is the volume element on $\hyp$.

We now take $\delta_{2}$ to correspond to a gauge transformation as in \cref{eq:Max-symm} to get
\begin{align}\label{maxwell-symp-current}
    \pb{\dd\omega}(\delta A,\delta_{\dd{\lambda}} A)=  - \dd\varepsilon_{3} \delta \dd E^a \dd D_a \dd{\lambda} 
    = - \dd\varepsilon_{3} \dd D^{a}(\delta \dd{E}_{a} \dd{\lambda})\,.
\end{align}
where in the last step we have used the linearized Maxwell equation $\dd D_{a} \delta \dd{E}^{a}=0$ (see \cref{eq:Max-eqn-asymp}). That is, the symplectic current (with one of the perturbations being generated by a gauge transformation) can be written as a total derivative of \(\delta\dd E_a \dd\lambda\). Thus we define the perturbed charge \(\delta \mc Q[\dd{\lambda}; S]\) on a cross-section \(S\) of \(\hyp\) by
\be
    \delta \mc Q[\dd{\lambda}; S] = \int_S \dd{\varepsilon}_{_2} \dd u^{a} \delta\dd{E}_{a} \dd{\lambda}\,,
\ee
where $\dd{\varepsilon}_{2} \equiv \dd{\varepsilon}_{ab}$ is the area element on $S$ and $\dd u^{a}$ is the future-directed normal to it. Note that this expression is manifestly integrable and defines the unperturbed charge once we choose a reference solution on which \(\mc Q[\dd{\lambda}; S] = 0 \) for all \(\dd\lambda\) and all \(S\). For the reference solution we choose the trivial solution \(F_{ab} = 0\) so that \(\dd E_a = 0\). Then the unperturbed charge is given by
\be \label{eq:Maxwell-charge}
    \mc Q[\dd{\lambda}; S] = \int_S \dd{\varepsilon}_{_2} \dd u^{a} \dd{E}_{a} \dd{\lambda}  \,,
\ee

Let \(\Delta\hyp\) be any region of \(\hyp\) bounded by the cross-sections \(S_2\) and \(S_1\) (with \(S_2\) in the future of \(S_1\)), then the flux of the charge \cref{eq:Maxwell-charge} through $\Delta\hyp$ is given by
\be \label{Maxwell-flux}
    \mc {F} [\dd{\lambda}; \Delta \mc H] = - \int_{\Delta\hyp} \dd{\varepsilon}_{_3} \dd{E}_{a} \dd D^{a} \dd{\lambda}   \,.
\ee
Note that the flux of the charge vanishes for $\dd{\lambda} = \text{constant}$ in which case \cref{eq:Maxwell-charge} is the Coulomb charge. The charges associated with a general smooth $\dd{\lambda}$ are only associated with the blowup \(\hyp\) and not to $i^{0}$ itself. These additional charges are nevertheless useful to relate the charges defined on past and future null infinity and derive the resulting conservation laws for their fluxes in a scattering process; see \cite{KP-EM-match}.

%%============================================================
\section{Gravitational fields and Einstein equations at \(i^0\)}
\label{sec:GR-eom}

Now we turn to a similar analysis of symmetries, charges and fluxes for general relativity. To set the stage in this section we analyze the consequences of Einstein equations and the universal structure common to all spacetimes satisfying \cref{def:AH}.

Using the conformal transformation relating the unphysical Ricci tensor \(R_{ab}\) to the physical Ricci tensor \(\hat R_{ab}\) (see Appendix~D of \cite{Wald-book}), the vacuum Einstein equation \(\hat R_{ab} = 0\) can be written as
\be\label{eq:EE}\begin{aligned}
    S_{ab} & = - 2 \Omega^{-1} \nabla_a \nabla_b \Omega + \Omega^{-2} \nabla^{c}\Omega \nabla_{c}\Omega g_{ab}\,, \\
    \Omega^{\half} S_{ab} & = -4 \nabla_a \eta_b + 4 \Omega^{-\half} \lb( g_{ab} - \tfrac{1}{\eta^{2}} \eta_{a} \eta_{b} \rb)\eta_c \eta^c \,,
\end{aligned}\ee
where, as before, \(\eta_a = \nabla_a \Omh\), and \(S_{ab}\) is given by
\be\label{eq:S-defn}
S_{ab} \defn R_{ab} - \tfrac{1}{6} R g_{ab}\,.
\ee
Further, the Bianchi identity \(\nabla_{[a} R_{bc]de} = 0\) on the unphysical Riemann tensor along with \cref{eq:EE} gives the following equations for the unphysical Weyl tensor \(C_{abcd}\) (see \cite{Geroch-asymp} for details).
\begin{subequations}\label{eq:Bianchi-unphys}\begin{align}
	\nabla_{[e} (\Omega^{-1} C_{ab]cd}) = 0 \label{eq:curl-weyl}\,, \\
	\nabla^d C_{abcd} = - \nabla_{[a} S_{b]c}\,. \label{eq:Weyl-S}
\end{align}\end{subequations}

Since the physical Ricci tensor $\hat{R}_{ab}$ vanishes, the gravitational field is completely described by the physical Weyl tensor $\hat{C}_{abcd}$. The unphysical Weyl tensor is then \(C_{abcd} = \Omega^2 \hat C_{abcd}\). Since the unphysical metric \(g_{ab}\) is \(C^{>0}\) at \(i^0\), \(\Omh C_{abcd}\) is \(C^{>-1}\) at \(i^0\) \cite{AH}, and let
\be
\dd{C}_{abcd}(\vec\eta) \defn \lim_{\to i^{0}} \Omega^{\half} C_{abcd}\,.
\ee
The \emph{electric} and \emph{magnetic} parts of \(\dd C_{abcd}(\vec\eta)\) are, respectively, defined by
\be \label{eq:EB-defn}
    \dd{E}_{ab}(\vec\eta) \defn \dd{C}_{acbd} (\vec\eta) \dd{\eta}^{c} \dd{\eta}^{d} \eqsp \dd{B}_{ab}(\vec\eta) \defn * \dd{C}_{acbd} (\vec\eta) \dd{\eta}^{c}\dd{\eta}^{d}\,.
\ee
where \(*\dd C_{abcd}(\vec\eta) \defn \tfrac{1}{2} \dd \varepsilon_{ab}{}^{ef} \dd C_{efcd} (\vec\eta)\). It follows from the symmetries of the Weyl tensor that both \(\dd E_{ab}(\vec\eta)\) and \(\dd B_{ab}(\vec\eta)\) are orthogonal to \(\dd\eta^a\), symmetric and traceless with the respect to the metric \(\dd h_{ab}\) on $\hyp$, and thus define smooth tensor fields \(\dd E_{ab}\) and \(\dd B_{ab}\) on \(\hyp\), respectively. The limiting Weyl tensor can be obtained from these fields using  
\be\label{eq:E-B-decomp} 
    \dd{C}^{ab}{}_{cd}(\vec\eta) = 4 \dd{\eta}^{[a} \dd{\eta}_{[c} \dd{E}^{b]}{}_{d]} - 4 \dd{h}^{[a}{}_{[c}\dd{E}^{b]}{}_{d]} + 2 \dd{\varepsilon}^{abe} \dd{\eta}_{[c}\dd{B}_{d]e} + 2 \dd{\varepsilon}_{cde}\dd{\eta}^{[a}\dd{B}^{b]e} \,.
\ee
Further, as shown in \cite{AH}, multiplying \cref{eq:curl-weyl} by \(\Omega\) and taking the limit to \(i^0\) gives the equations of motion
\be \label{eq:EB-curl}
\dd D_{[a} \dd E_{b]c} =0 \eqsp \dd D_{[a} \dd B_{b]c}=0\,.
\ee
These are the asymptotic Einstein equations at spatial infinity. Taking the trace over the indices \(a\) and \(c\) and using the fact that $\dd{E}_{ab}$ and $\dd{B}_{ab}$ are traceless, it also follows that
\be\label{eq:EB-div}
    \dd D^b \dd{E}_{ab} = \dd D^b \dd{B}_{ab}= 0 \,.
\ee

To apply the symplectic formalism to general relativity, we will need to consider metric perturbations instead of just perturbations of the Weyl tensor. As we will show below (\cref{eq:gamma-E-K}) suitably rescaled limits of the unphysical metric perturbations can be expressed in terms of perturbations of certain potentials for \(\dd E_{ab}\) and \(\dd B_{ab}\) provided by the tensor \(S_{ab}\) in \cref{eq:S-defn}. These potentials are obtained as follows: Since \(g_{ab}\) is \(C^{>0}\), \(\Omh S_{ab}\) is \(C^{>-1}\) and let \(\dd S_{ab} (\vec\eta) \defn \lim_{\to i^0} \Omh S_{ab}\). Define
\be \label{eq:potentials-defn}
    \dd{E}(\vec\eta) \defn \dd{S}_{ab}(\vec\eta) \dd{\eta}^{a}\dd{\eta}^{b} \eqsp \dd{K}_{ab}(\vec\eta) \defn \dd{h}_a{}^{c} \dd{h}_b{}^{d} \dd{S}_{cd}(\vec\eta) - \dd{h}_{ab} \dd{E}(\vec\eta)\,, 
\ee
which induce the fields \(\dd E\) and \(\dd K_{ab}\) intrinsic to \(\hyp\). Following \cite{AH}, multiplying \cref{eq:Weyl-S} by \(\Omega\) and taking the limit to \(i^0\), along with \cref{eq:EB-curl} implies that
\be\label{eq:h-eta-S}
    \dd h_a{}^b \dd \eta^c \dd S_{bc}(\vec\eta) = \dd D_a \dd E\,,
\ee
and
\be \label{eq:EB-potentials}
    \dd{E}_{ab} = -\tfrac{1}{4} (\dd D_{a}\dd D_{b}\dd{E} +  \dd h_{ab} \dd{E}) \eqsp \dd{B}_{ab} = -\tfrac{1}{4}\dd{\varepsilon}_{cda}\dd D^{c}\dd{K}^{d}{}_{b}\,.
\ee
Thus, \(\dd E\) is a scalar potential for \(\dd E_{ab}\) while \(\dd K_{ab}\) is a tensor potential for \(\dd B_{ab}\).\footnote{Since \(\dd B_{ab}\) is curl-free (\cref{eq:EB-curl}), there also exists a scalar potential for \(\dd B_{ab}\) (see \cref{lem:scalar-pot}). However this scalar potential cannot be obtained as the limit of some tensor field on spacetime.}

The potentials \(\dd E\) and \(\dd K_{ab}\) are not free fields on \(\hyp\). Suitably commuting the derivatives and using \cref{eq:Riem-hyp} one can verify that \(\dd E_{ab}\) identically satisfies \cref{eq:EB-curl} when written in terms of the potential \(\dd E\) while \(\dd h^{ab} \dd E_{ab} = 0\) gives
\be \label{eq:box-E}
    \dd D^{2}\dd E + 3 \dd{E} = 0\,.
\ee
On the other hand, since \(\dd K_{ab}\) is symmetric the magnetic field \(\dd B_{ab}\) in \cref{eq:EB-potentials} is identically traceless. Since \(\dd B_{ab}\) is symmetric and satisfies \cref{eq:EB-curl}, we get that
\begin{subequations}\label{eq:K-eom}\begin{align}
    \dd{\varepsilon}_a{}^{bc} \dd{B}_{bc} = 0 \implies \dd D^{b}\dd{K}_{ab} & = \dd D_{a} \dd{K} \label{eq:div-K}\,, \\
    \dd\varepsilon_a{}^{cd} \dd D_c \dd{B}_{db} = 0 \implies \dd D^{2}\dd{K}_{ab} & = \dd D_{a}\dd D_{b} \dd{K} + 3 \dd{K}_{ab} - \dd{h}_{ab} \dd{K} \label{eq:box-K}\,,
\end{align}\end{subequations}
where \(\dd K \defn \dd h^{ab} \dd K_{ab}\), and to get \cref{eq:box-K} we have commuted derivatives using \cref{eq:Riem-hyp} and used \cref{eq:div-K}. Considering the potentials \(\dd E\) and \(\dd K_{ab}\) as the basic fields, the asymptotic Einstein equations are given by \cref{eq:box-E,eq:K-eom}, while the Weyl tensors \(\dd E_{ab}\) and \(\dd B_{ab}\) are derived quantities through \cref{eq:EB-potentials}.\\

To define the charge for asymptotic Lorentz symmetries, e.g. angular momentum in \cref{sec:lorentz-charge}, we will need the ``subleading'' part of the magnetic Weyl tensor. Following Ashtekar and Hansen \cite{AH}, we will restrict to the class of spacetimes satisfying the additional condition \(\dd B_{ab} = 0\). We also require that the ``subleading'' magnetic field defined by
\be \label{eq:beta-defn}
    \dd{\beta}_{ab} \defn \lim_{\to i^{0}} * C_{acbd} \eta^{c} \eta^{d}\,,
\ee
exists as a \(C^{>-1}\) tensor field at \(i^0\). The condition \(\dd B_{ab} = 0\) is satisfied in any spacetime which is \emph{either} stationary \emph{or} axisymmetric \cite{B-zero}. In \cref{sec:new-beta} we show how one can define a ``subleading'' magnetic Weyl tensor and the Lorentz charges even when \(\dd B_{ab} \neq 0\). Since those computations are more tedious we impose the above restriction in the main body of the paper.

The consequences of this restriction are as follows. Since \(\dd B_{ab} = 0\) from \cref{eq:EB-potentials} the ``curl'' of $\dd{K}_{ab}$ vanishes
\be \label{eq:vanishing-curl-K}
    \dd D_{[a} \dd{K}_{b]c} = 0\,.
\ee
It follows from \cref{lem:scalar-pot} that there exists a scalar potential \(\dd k\) such that
\be\label{eq:K-potential}
    \dd{K}_{ab} = \dd D_{a}\dd D_{b} \dd{k} + \dd{h}_{ab}\dd{k}\,.
\ee
The scalar potential \(\dd k\) is a free function on \(\hyp\) since the equations of motion \cref{eq:K-eom} are identically satisfied after using \cref{eq:K-potential}. Using the freedom in the conformal factor one can now set \(\dd K_{ab} = 0\) (see \cite{AH} and \cref{rem:GR-gauge-choice}). Since, we do not wish to impose any restrictions on the conformal factor, we will \emph{not} demand that \(\dd K_{ab}\) vanishes.

Note that from \cref{eq:beta-defn} it follows that \(\dd\beta_{ab}\) is symmetric, tangent to \(\hyp\) and traceless. In the following we shall also need an equation of motion for \(\dd\beta_{ab}\) which is obtained as follows: Contract the indices \(e\) and \(d\) in \cref{eq:curl-weyl} and multiply by 
\(3 \Omega\) to get
\be
    \nabla^d C_{abcd} = \Omega^{-1} C_{abcd}\nabla^d \Omega = 2 \Omega^{-\half} C_{abcd}\eta^d\,.
\ee
Using the Hodge dual of the above equation we obtain
\be
    \Omh \nabla^b (*C_{acbd} \eta^c \eta^d) = - 2 *C_{acbd}\eta^b \eta^c\eta^d + 2 \Omh *C_{acbd} \nabla^b \eta^{(c} \eta^{d)}\,.
\ee
The first term on the right-hand-side vanishes due to the symmetries of the Weyl tensor. In the second term on the right-hand-side we substitute for the derivative of \(\eta_a\) using \cref{eq:EE} to get
\be
    \Omh \nabla^b (*C_{acbd} \eta^c \eta^d) = - \tfrac{1}{4} (\Omh *C_{acbd}) (\Omh S^{bc}) \eta^d\,.
\ee
Taking the limit to \(i^0\), writing the tensor \(\dd S_{ab}\) in terms of the gravitational potentials through \cref{eq:potentials-defn,eq:h-eta-S}, and using \(\dd B_{ab} = 0\) along with \cref{eq:E-B-decomp}, we get the equation of motion
\be\label{eq:div-beta}
    \dd D^b \dd\beta_{ab} =  \tfrac{1}{4} \dd\varepsilon_{cda} \dd E^c{}_b \dd K^{bd}\,.
\ee

\begin{remark}[Conformal transformations of the asymptotic fields]\label{rem:conf-GR-fields}
    Under changes of the conformal factor \(\Omega \mapsto \omega\Omega\) we have
\be\begin{aligned}
	S_{ab} & \mapsto S_{ab} - 2 \omega^{-1} \nabla_a \nabla_b \omega + 4 \omega^{-2} \nabla_a \omega \nabla_b \omega - \omega^{-2} g_{ab} \nabla^c \omega \nabla_c \omega\,, \\
	C_{abcd} & \mapsto \omega^2 C_{abcd}\,. 
\end{aligned}\ee
From the conditions in \cref{rem:conf} it follows that \(\dd E_{ab}\), \(\dd B_{ab}\) and \(\dd E\) are invariant while
\be\label{eq:conf-K}
    \dd K_{ab} \mapsto \dd K_{ab} - 2 (\dd D_a \dd D_b \dd\alpha + \dd h_{ab}\dd\alpha)\,.
\ee
Further, when \(\dd B_{ab} = 0\) we also have the transformation of the ``subleading'' magnetic Weyl tensor \(\dd\beta_{ab}\) given by
\be\label{eq:conf-beta}
    \dd\beta_{ab} \mapsto \dd\beta_{ab} - \dd\varepsilon_{cd(a} \dd E^c{}_{b)} \dd D^d \dd\alpha\,.
\ee 

\end{remark}
 
%%----------------------------------------------------------
\subsection{The universal structure at \(i^0\)}
\label{sec:univ-str}

In this section we summarize the \emph{universal structure} at \(i^0\), that is, the structure common to all spacetimes which are asymptotically-flat in the sense of \cref{def:AH} and thus is independent of the choice of the physical spacetime under consideration.

Consider any two unphysical spacetimes \((M, g_{ab}, \Omega)\) and \((M', g'_{ab},\Omega')\) with their respective \(C^{>1}\) differential structures at their spatial infinities corresponding to two different physical spacetimes. Using a \(C^1\) diffeomorphism we can identify the points representing the spatial infinities and their tangent spaces without any loss of generality. Each of the metrics \(g_{ab}\) and \(g'_{ab}\) induces a metric in the tangent space \(Ti^0\) which is isometric to the Minkowski metric. Thus, the metric \(\dd g_{ab}\) at \(i^0\) is also universal. This also implies that the spatial directions \(\vec\eta\), the space of directions \(\hyp\) and the induced metric \(\dd h_{ab}\) are universal.

So far we have only used the \(C^1\) differential structure. However since the differential structure at \(i^0\) is slightly better, being \(C^{>1}\), we can identify the spacetimes at the ``next order''. In \cite{AH} this structure was imposed by suitably identifying spacelike geodesics in the \emph{physical} spacetimes. But as pointed out by \cite{Porrill} this identification cannot be performed except in very special cases. Below we argue that a similar identification of the spacetimes can be done using equivalence classes of \(C^{>1}\) curves in the unphysical spacetimes. The proof is based on constructing a suitable \(C^{>1}\) coordinate system at \(i^0\) and is deferred to \cref{sec:coord}, we summarize the main construction below.

Consider the unphysical spacetime \((M, g_{ab}, \Omega)\), and a spacelike \(C^{>1}\) curve \(\Gamma_v\) in \(M\) passing through \(i^0\) with tangent \(v^a\). Since the curve is \(C^{>1}\) its tangent vector \(v^a\) is \(C^{>0}\). Using the universal metric \(\dd g_{ab}\) at \(i^0\) we can then demand that \(v^a\) be unit-normalized at \(i^0\) and thus along the curve \(\Gamma_v\)
\be
    \lim_{\to i^0} v^a = \dd\eta^a\,,
\ee
that is the curve \(\Gamma_v\) points in some spatial direction \(\vec\eta\) at \(i^0\). Further, since \(\Gamma_v\) is \(C^{>1}\),  \(v^b \nabla_b v^a\) is a \(C^{>-1}\) vector. Thus, define the \emph{acceleration} of \(\Gamma_v\) at \(i^0\) by the projection of this vector on to \(\hyp\)
\be\label{eq:acc-defn}
    \dd A^a[\Gamma_v] \defn \dd h^a{}_b \lim_{\to i^0} v^c \nabla_c v^b\,.
\ee

Now we define the curves \(\Gamma_v\) (with tangent \(v^a\)) and \(\Gamma_\eta\) (with tangent \(\eta^a\)) to be equivalent if their accelerations are equal at \(i^0\). To see what this entails, note that since \(v^a\) is \(C^{>0}\) and equals \(\dd\eta^a\) in the limit to \(i^0\) we have that \(v^a = \eta^a + \Omh w^a\) for some \(w^a\) which is \(C^{>-1}\) at \(i^0\). Then, from \cref{eq:acc-defn} we have
\be
    \dd A^a[\Gamma_v] = \dd A^a[\Gamma_\eta] \iff \dd h_{ab} \lim_{\to i^0} w^b = 0\,.
\ee
Thus, we have an equivalence class of curves through \(i^0\) pointing in each direction \(\vec\eta\) defined by\footnote{These equivalence classes of curves form a principal bundle over \(\hyp\), called \emph{Spi} in \cite{AH}.}
\be\label{eq:equiv-curves}
    \Gamma_v \sim \Gamma_\eta \iff \dd h_{ab}\lim_{\to i^0} \Omega^{-\half} (v^b - \eta^b) = 0\,.
\ee
We will show in \cref{sec:coord} that using a \(C^{>1}\) diffeomorphism one can identify these equivalence classes of curves between any any two spacetimes \((M, g_{ab}, \Omega)\) and \((M', g'_{ab}, \Omega')\). Further, we show that the conformal factors \(\Omega\) and \(\Omega'\) can also be identified in a neighbourhood of \(i^0\).

Thus, the universal structure at \(i^0\) consists of the point \(i^0\), the tangent space \(Ti^0\), the metric \(\dd g_{ab}\) at \(i^0\) and the equivalence classes of \(C^{>1}\) curves given by \cref{eq:equiv-curves}. In addition, the conformal factor \(\Omega\) can also be chosen to be universal.

\begin{remark}[Logarithmic translations]
\label{rem:log-trans}
So far we have worked with a fixed \(C^{>1}\) differential structure in the unphysical spacetime at \(i^0\). But given a physical spacetime the unphysical spacetime is ambiguous up to a \(4\)-parameter family of \emph{logarithmic translations} at \(i^0\) which simultaneously change the \(C^{>1}\) differential structure and the conformal factor at \(i^0\); see \cite{Ash-log} or Remark~B.1 of \cite{KP-GR-match} for details. The logarithmic translations at \(i^0\) are parameterized by a \emph{direction-independent} vector \(\dd\Lambda^a\) at \(i^0\). Any such vector can be written as 
\be
    \dd \Lambda^a = \dd\Lambda \dd\eta^a + \dd D^a \dd\Lambda\,,
\ee
where \(\dd\Lambda(\vec\eta) = \dd\eta_a \dd\Lambda^a\) is a function on \(\hyp\) satisfying
\be\label{eq:log-trans-eqn}
    \dd D_a \dd D_b \dd\Lambda + \dd h_{ab} \dd\Lambda = 0\,.
\ee
Under such logarithmic translations the potentials \cref{eq:potentials-defn} transform as \cite{Ash-log}
\be\label{eq:log-trans}
    \dd E \mapsto \dd E + 4\dd\Lambda \eqsp \dd K_{ab} \mapsto \dd K_{ab}\,,
\ee
while \(\dd E_{ab}\) and \(\dd B_{ab}\) are invariant. The presence of these logarithmic translations will lead to the following issue when we define the charges for supertranslations in \cref{sec:st}. For general supertranslations (which are not translations) our charges will depend on the potential \(\dd E\) instead of just the electric field \(\dd E_{ab}\). Thus, even if we take the physical spacetime to be the Minkowski spacetime our charges will not vanish due to the logarithmic translation ambiguity \cref{eq:log-trans} in \(\dd E\). Thus, now we will fix these logarithmic translations following the argument in \cite{Ash-log}.

Since the metric \(\dd g_{ab}\) in the tangent space \(Ti^0\) is universal and isometric to the Minkowski metric it is invariant under the reflection of the spatial directions \(\vec\eta \mapsto - \vec\eta\). This gives rise to a reflection isometry of the metric \(\dd h_{ab}\) on the space of directions \(\hyp\). Now it was shown in \cite{KP-GR-match} that the only spacetimes which are asymptotically-flat at spatial infinity and which ``match'' on to asymptotically-flat spacetimes on null infinity are the ones where \(\dd E_{ab}\) is reflection-even, i.e.
\be
    \dd E_{ab} (\vec\eta) = \dd E_{ab}(-\vec\eta)\,.
\ee
Further, since \(\dd\Lambda = \dd\eta_a \dd\Lambda^a\) for the \emph{direction-independent} vector \(\dd\Lambda^a\) we have that, \(\dd\Lambda\) is reflection-odd
\be
    \dd \Lambda(\vec\eta) = - \dd\Lambda(-\vec\eta)\,.
\ee
For a reflection-even \(\dd E_{ab}\), from \cref{eq:EB-potentials,eq:log-trans-eqn}, it follows that using a logarithmic translation we can demand that the potential \(\dd E\) is also reflection-even, so that 
\be\label{eq:parity-E}
    \dd E (\vec\eta) = \dd E(-\vec\eta)\,.
\ee
Having fixed the logarithmic translations in this way, \(\dd E_{ab} = 0\) then implies that \(\dd E = 0\). In particular, for Minkowski spacetime we have
\be\label{eq:Mink-stuff}
    \dd E = 0 \eqsp \dd B_{ab} = 0 \eqsp \dd \beta_{ab} = 0 \quad \text{(on Minkowski spacetime)}\,.
\ee
Note that when \(\dd E_{ab} = 0\), \(\dd \beta_{ab}\) is conformally-invariant (see \cref{eq:conf-beta}) and the conditions \cref{eq:Mink-stuff} do not depend on the conformal factor chosen for Minkowski spacetime. These conditions will ensure that our all our charges will vanish on Minkowski spacetime. Thus, from here on we will assume that the logarithmic translations have been fixed as above that is, we work the choice of \(C^{>1}\) differential structure at \(i^0\) where the parity condition \cref{eq:parity-E} is satisfied.
\end{remark}

%%=============================================================
\section{Metric perturbations and symplectic current at \(i^0\)}
\label{sec:pert}
Now consider a one-parameter family of asymptotically-flat physical metrics \(\hat g_{ab}(\lambda)\) where \(\hat g_{ab} = \hat g_{ab}(\lambda = 0)\) is some chosen background spacetime. Define the physical metric perturbation \(\hat \gamma_{ab}\) around the background \(\hat g_{ab}\) by
\be\label{eq:phys-pert}
    \hat \gamma_{ab} = \delta \hat g_{ab} \defn \lb. \frac{d}{d\lambda} \hat g_{ab}(\lambda) \rb\vert_{\lambda = 0}\,.
\ee
We will use ``\(\delta\)'' to denote perturbations of other quantities defined in a similar way.

As discussed above, the conformal factor \(\Omega\) can be chosen universally, i.e., independently of the choice of the physical metric. The unphysical metric perturbation is
\be\label{ond}
	\delta g_{ab} = \gamma_{ab} = \Omega^2 \hat \gamma_{ab}\,,
\ee
and we also have
\be \label{eq:variations-1}
\delta \eta_{a} = \delta \nabla_a \Omh = 0 \eqsp \delta \eta^{a} = \delta(g^{ab}\eta_b) = -\gamma^{ab} \eta_{b}\,.
\ee

Now we investigate the conditions on the unphysical perturbation \(\gamma_{ab}\) which preserve asymptotic flatness and the universal structure at \(i^0\) described in \cref{sec:univ-str}. First recall that since the unphysical metric \(g_{ab}\) is $C^{>0}$ and universal at $i^{0}$, it follows that the unphysical metric perturbation \(\gamma_{ab}\) is \(C^{>0}\) and \(\gamma_{ab}\vert_{i^0} = 0\). Therefore 
\be \label{eq:lim-gamma}
    \dd\gamma_{ab}(\vec\eta) \defn \lim_{\to i^0} \Omega^{-\half} \gamma_{ab} \text{ is } C^{>-1}\,,
\ee
With \cref{eq:variations-1,eq:lim-gamma} we also see that \(\delta \dd\eta^a = 0\). Thus, the metric perturbation also preserves the spatial directions \(\vec\eta\) at \(i^0\), the space of directions \(\hyp\) and the metric \(\dd h_{ab}\) on it.

Now consider the universal structure given by the equivalence classes of \(C^{>1}\) curves through \(i^0\) as described in \cref{sec:univ-str}. Consider the equivalence class of a fixed curve \(\Gamma_v\) with tangent \(v^a\). For this equivalence class to be preserved, the perturbation of \cref{eq:equiv-curves} must vanish. Evaluating this condition using \cref{eq:variations-1,eq:lim-gamma} we obtain the condition
\be\label{eq:gamma-eta-h}
    \dd h_a{}^b \dd \eta^c \dd \gamma_{bc}(\vec\eta) = 0\,.
\ee
In summary, \cref{eq:lim-gamma,eq:gamma-eta-h} are the asymptotic conditions on the unphysical metric perturbations which preserve the asymptotic flatness and the universal structure at \(i^0\).

The metric perturbation \(\dd\gamma_{ab}\) can be directly related to the perturbations of the gravitational potentials \(\dd E\) and \(\dd K_{ab}\) defined in \cref{eq:potentials-defn}. Perturbing \cref{eq:EE} to evaluate $\Omega^{\half} \delta S_{ab}$ and taking the limit to \(i^0\) using \cref{eq:variations-1,eq:lim-gamma} we get
\be
    \delta\dd{S}_{ab} = \lim_{\to i^{0}} \Omega^{\half}\delta S_{ab} = 4 \dd{\partial}_{(a} \dd{\gamma}_{b)c} \dd{\eta}^{c} + 4\dd{\eta}_{(a} \dd{\gamma}_{b)c}\dd{\eta}^{c} + 2\dd{\gamma}_{ab} - 4 \dd{\gamma}_{cd} \dd{\eta}^{c}\dd{\eta}^{d} \dd g_{ab}\,.
\ee
Using the definition of the gravitational potentials \cref{eq:potentials-defn} and \cref{eq:gamma-eta-h} we obtain
\begin{subequations}\label{eq:deltaEK}\begin{align} 
    \delta \dd{E} & = 2 \dd{\gamma}_{ab}\dd{\eta}^{a} \dd{\eta}^{b} \label{eq:deltaE}\,, \\
    \delta \dd{K}_{ab} & = -2 \dd{h}_a{}^{c} \dd{h}_b{}^{d} \dd{\gamma}_{cd} - \dd{h}_{ab}\delta \dd{E} \label{eq:deltaK}\,.
\end{align}\end{subequations}
Using \cref{eq:gamma-eta-h,eq:deltaEK} we can reconstruct the metric perturbation \(\dd\gamma_{ab}(\vec\eta)\) in terms of the perturbed gravitational potentials on \(\hyp\) as
\be \label{eq:gamma-E-K}
    \dd\gamma_{ab}(\vec\eta) = \tfrac{1}{2} \lb[ \delta \dd{E} (\dd{\eta}_{a}\dd{\eta}_{b} - \dd h_{ab}) - \delta \dd{K}_{ab}  \rb]\,.
\ee
The linearized Einstein equations for \(\dd \gamma_{ab}\) in the form \cref{eq:gamma-E-K} are then equivalent to the linearizations of \cref{eq:box-E,eq:K-eom}.\\

Next we consider the behaviour of the symplectic current of vacuum general relativity near \(i^0\). The symplectic current is given by (see \cite{WZ})
\be\label{eq:sympcurrent}
	\omega_{abc} = - \tfrac{1}{16 \pi}\hat\varepsilon_{abcd} \hat w^d \quad\text{with}\quad \hat w^a = \hat P^{abcdef} \hat\gamma_{2bc} \hat\nabla_d \hat\gamma_{1ef} - [1 \lra 2]\,,
\ee
where ``\([1 \lra 2]\)'' denotes the preceding expression with the \(1\) and \(2\), labeling the perturbations, interchanged and the tensor \(\hat P^{abcdef}\) is given by
\be\label{eq:P-defn}
    \hat P^{abcdef} = \hat{g}^{ae}\hat{g}^{fb}\hat{g}^{cd} - \tfrac{1}{2}\hat{g}^{ad}\hat{g}^{be}\hat{g}^{fc} - \tfrac{1}{2}\hat{g}^{ab}\hat{g}^{cd}\hat{g}^{ef} - \tfrac{1}{2}\hat{g}^{bc}\hat{g}^{ae}\hat{g}^{fd} + \tfrac{1}{2}\hat{g}^{bc}\hat{g}^{ad}\hat{g}^{ef}\,.
\ee

To analyse the behaviour of the symplectic current in the limit to \(i^0\) we first express it in terms of quantities in the unphysical spacetime using
\be
    \varepsilon_{abcd} = \Omega^{4} \hat{\varepsilon}_{abcd} \eqsp P^{abcdef} = \Omega^{-6} \hat{P}^{abcdef} \eqsp  \gamma_{ab} = \Omega^{2} \hat{\gamma}_{ab}\,,
\ee
where \(P^{abcdef}\) is defined through the unphysical metric by the same expression as \cref{eq:P-defn}. Using these, and converting the physical derivative operator \(\hat\nabla\) to the unphysical one \(\nabla\) as
\begin{equation}
\hat{\nabla}_{d} \hat{\gamma}_{1 ef} = \nabla_{d} \hat{\gamma}_{1 ef} + \Omega^{-1} [\hat{\nabla}_{d} \Omega \hat{\gamma}_{1 ef} + \hat{\nabla}_{e} \Omega \hat{\gamma}_{1 df} - g_{ed} \hat{\nabla}^{a} \Omega \hat{\gamma}_{1 af}  + (e \leftrightarrow f)]\,,
\end{equation}
we obtain 
\be\begin{aligned}\label{eq:sympcurrent-unphys}
	\omega_{abc} & = - \tfrac{1}{16 \pi}\varepsilon_{abcd} w^d \,, \quad \\[1.5ex]
	\text{with}\quad 
    w^a & = \Omega^{-2} P^{abcdef} \gamma_2{}_{bc} \nabla_d \gamma_1{}_{ef} + \Omega^{-3} \gamma_1^{ab} \nabla_b \Omega \gamma_{2 c}{}^c - [1 \lra 2] \,.
\end{aligned}\ee

Converting to quantities which are direction-dependent at \(i^0\) and using \cref{eq:lim-gamma} we see that \(\Omega^{\nfrac{3}{2}} \omega_{abc}\) is \(C^{>-1}\). The pullback \(\pb{\dd\omega}\) to \(\hyp\) of \(\lim_{\to i^0}\Omega^{\nfrac{3}{2}} \omega_{abc}\) is given by
\be\label{eq:sympcurrent-H}
    \pb{\dd\omega} = - \frac{1}{16\pi} \dd\varepsilon_{3}~ \dd\eta^a \lb( 2 \dd{\eta}^{b} \dd\gamma_{2 ab} \dd\gamma_{1} - \tfrac{1}{2} \dd\gamma_{1 ab} \dd{\partial}^{b} \dd\gamma_{2} + \dd\gamma_{1}^{bc} \dd{\partial}_{c} \dd\gamma_{2 ab} - \tfrac{1}{2} \dd\gamma_{1} \dd{\partial}^{b} \dd\gamma_{2 ab} \rb) - [1 \lra 2]\,.
\ee
This expression can be considerably simplified by rewriting it in terms of the perturbed gravitational potentials \(\delta \dd E\) and \(\delta \dd K_{ab}\) using \cref{eq:gamma-E-K}. An easy but long computation gives
\be\label{eq:sympcurrent-H-simplified}
    \pb{\dd\omega} =  \frac{1}{64 \pi}\dd\varepsilon_3~ (\delta_{1} \dd{K} \delta_{2} \dd{E} - \delta_{2} \dd{K} \delta_{1} \dd{E})\,,
\ee
where, as before, $\dd{K}\defn \dd{h}^{ab} \dd{K}_{ab}$.

%%==========================================================
\section{Asymptotic symmetries at \(i^0\): The \(\spi\) algebra}
\label{sec:spi-symm}

In this section we analyze the asymptotic symmetries at \(i^0\). We show that the diffeomorphisms of the physical spacetime which preserve the asymptotic flatness of the spacetime (defined by \cref{def:AH}) generate an infinite-dimensional algebra \(\spi\). This asymptotic symmetry algebra was obtained in \cite{AH,Ash-in-Held} by analyzing the infinitesimal diffeomorphisms which preserve the universal structure at \(i^0\). Here we provide an alternative derivation by considering the physical perturbations generated by such infinitesimal diffeomorphisms and demanding that the corresponding unphysical perturbations satisfy the asymptotic conditions \cref{eq:lim-gamma,eq:gamma-eta-h}.

Consider an infinitesimal diffeomorphism generated by a vector field \(\hat \xi^a\) in the physical spacetime, and let \(\xi^a = \hat\xi^a\)  be the corresponding vector field in the unphysical spacetime. For \(\xi^a\) to be a representative of an asymptotic symmetry at \(i^0\) the infinitesimal diffeomorphism generated by \(\xi^a\) must preserve the universal structure at $i^{0}$. Firstly, the infinitesimal diffeomorphism must keep the the point $i^{0}$ fixed and preserve the $C^{>1}$ differential structure at $i^{0}$. Thus, \(\xi^a\) must be $C^{>0}$ at \(i^0\) and \(\xi^a \vert_{i^0} = 0\). This implies that \(\Omega^{-\half} \xi^a\) is \(C^{>-1}\) at \(i^0\) and let
\be \label{eq:X-defn}
    \dd X^a(\vec\eta) \defn \lim_{\to i^{0}}\Omega^{-\half} \xi^{a}\,.
\ee

Now consider the physical metric perturbation $ \hat\gamma^{(\xi)}_{ab} = \delta_\xi \hat{g}_{ab} \defn \lie_{\xi} \hat{g}_{ab}$ corresponding to an infinitesimal diffeomorphism generated by $\xi^{a}$. The corresponding unphysical metric perturbation is given by
\begin{align} \label{eq:lin-diffeo}
    \gamma^{(\xi)}_{ab} = \Omega^{2} \lie_{\xi} \hat g_{ab} = \lie_{\xi} g_{ab} - 4\Omega^{-\half} \xi^{c} \eta_{c} g_{ab}\,.
\end{align}

Since \(\gamma^{(\xi)}_{ab}\) must satisfy the asymptotic conditions at \(i^0\) in \cref{eq:lim-gamma,eq:gamma-eta-h}, we have that \(\gamma^{(\xi)}_{ab}\) is \(C^{>0}\) at \(i^0\) and \(\gamma^{(\xi)}_{ab} \vert_{i^0} = 0\). To see the implications of these conditions first evaluate the condition \(\gamma^{(\xi)}_{ab} \vert_{i^0} = 0\) using \cref{eq:lin-diffeo,eq:X-defn} which gives
\be\label{eq:lorentz-cond}
    \dd{\eta}_{a}\dd{X}^{a}(\vec\eta) = 0 \eqsp \dd D_{(a} \dd X_{b)} = 0\,,
\ee
that is, the vector field \(\dd X^a\) is tangent to \(\hyp\) and is a Killing vector field on it. Thus, \(\dd X^a\) is an element of the Lorentz algebra \(\mf{so}(1,3)\). Some useful properties of these Killing vectors and their relationship to infinitesimal Lorentz transformations in the tangent space \(Ti^0\) are collected in \cref{sec:Killing-H}.

Further, since both \(\gamma^{(\xi)}_{ab}\) and \(\lie_\xi g_{ab}\) are \(C^{>0}\) we must have that \(\Omega^{-\half} \xi^a \eta_a\) is also \(C^{>0}\). Since \(\Omega^{-\half} \xi^a \eta_a \vert_{i^0} = 0\) (which follows from \cref{eq:X-defn,eq:lorentz-cond}) we have that \(\Omega^{-1}\xi^a \eta_a\) is \(C^{>-1}\) at \(i^0\) so define
\be\label{eq:f-defn}
    \dd f(\vec\eta) \defn \lim_{\to i^0} \Omega^{-1}\xi^a \eta_a\,.
\ee
The function \(\dd f\) on \(\hyp\) then parametrizes the \emph{supertranslations}. The vector field generating a supertranslation can be obtained as follows. Consider \(\xi^a\) such that the corresponding \(\dd X^a\) (\cref{eq:X-defn}) vanishes and \(\dd\chi^a \defn \lim_{\to i^0} \Omega^{-1}\xi^a\) is \(C^{>-1}\) so that \(\dd f = \dd \chi^a \dd \eta_a\). Now consider the metric perturbation \cref{eq:lin-diffeo} corresponding to such a vector field. From \cref{eq:gamma-eta-h} we must have
\be
    \dd h_a{}^b \dd \eta^c \dd\gamma^{(\xi)}_{bc} = 0\,,
\ee 
where, as before, \(\dd\gamma^{(\xi)}_{ab} = \lim_{\to i^0} \Omega^{-\half} \gamma^{(\xi)}_{ab}\). Evaluating this condition using \cref{eq:lin-diffeo} and \(\dd\chi^a = \lim_{\to i^0} \Omega^{-1}\xi^a\) we get
\be
    \dd h_{ab} \dd\chi^b = - \dd D_a \dd f\,.
\ee
Thus a pure supertranslation \(\dd f\) is represented by a vector field \(\xi^a\) such that
\be \label{eq:supetr-vec-field}
    \lim_{\to i^{0}}\Omega^{-1} \xi^{a} = \dd{f} \dd\eta^{a} -  \dd D^{a} \dd{f}\,.
\ee

In summary, the asymptotic symmetries at \(i^0\) are parameterized by a pair \((\dd  f, \dd X^a)\) where \(\dd f\) is a smooth function and \(\dd X^a \in \mf{so}(1,3)\) is a smooth Killing vector field on \(\hyp\).\\

The Lie algebra structure of these symmetries can be obtained as follows. Let \(\xi^a_1\) and \(\xi^a_2\) be the vector fields representing the asymptotic Spi-symmetries \((\dd f_1, \dd X^a_1)\) and \((\dd f_2, \dd X^a_2)\) respectively. Then the Lie bracket \([\xi_1, \xi_2]^a = \xi^b_1 \nabla_b \xi^a_2 - \xi^b_2 \nabla_b \xi^a_1 \) of the representatives induces a Lie bracket on the Spi-symmetries. Using \cref{eq:X-defn,eq:lorentz-cond,eq:f-defn} the induced Lie bracket on the Spi-symmetries can be computed to be
\be\label{eq:spi-bracket}\begin{aligned}
    (\dd f, \dd X^a) &= [ (\dd f_1, \dd X^a_1), (\dd f_2, \dd X^a_2) ]\,, \\
    \text{with}\quad
    \dd f & = \dd X_1^b \dd D_b \dd f_2 - \dd X_2^b \dd D_b \dd f_1\,, \\
    \dd X^a &= \dd X_1^b \dd D_b \dd X_2^a - \dd X_2^b \dd D_b \dd X_1^a\,.
\end{aligned}\ee
Thus, the Spi symmetries form a Lie algebra \(\spi\) with the above Lie bracket structure. Note that if \(\dd X^a_1 = \dd X^a_2 = 0\) then \(\dd f = \dd X^a = 0\) ---  the supertranslations form an infinite-dimensional abelian subalgebra \(\mf s\). Further if \(\dd X^a_1 = 0\) and \(\dd X^a_2 \neq 0\) then \(\dd X^a = 0\), thus the supertranslations \(\mf s\) are a Lie ideal in \(\spi\). The quotient algebra \(\spi/\mf s\) is then isomorphic to the algebra of Killing fields on \(\hyp\) i.e. the Lorentz algebra \(\mf{so}(1,3)\). Thus the Spi symmetry algebra has the structure of a semi-direct sum
\be\label{eq:spi-semi-direct}
    \spi \cong \mf{so}(1,3) \ltimes \mf s\,.
\ee

The \(\spi\) algebra also has a preferred \(4\)-dimensional subalgebra \(\mf t \) of \emph{translations}. These are obtained as the supertranslations \(\dd f\) satisfying the additional condition
\be\label{eq:trans-cond}
    \dd D_a \dd D_b \dd f + \dd h_{ab} \dd f = 0\,.
\ee
The space of solutions to the above condition is indeed \(4\)-dimensional --- this can be seen from the argument in \cref{rem:trans-vectors} below, or by solving the equation in a suitable coordinate system on \(\hyp\); see Eqs.~D.204 and D.205 of \cite{CD} or Eq.~C.12 of \cite{KP-GR-match}. Further from \cref{eq:spi-bracket} it can be verified that the Lie bracket of a translation with any other element of \(\spi\) is again a translation, that is, the translations \(\mf t\) are a \(4\)-dimensional Lie ideal of \(\spi\).

\begin{remark}[Translation vectors at \(i^0\)]\label{rem:trans-vectors}

Let \(\dd v^a\) be a direction-independent vector at \(i^0\), and \(\dd v^a = \dd f \dd \eta^a + \dd f^a\) where \(\dd \eta_a \dd f^a = 0\). Then, since \(\dd v^a\) is direction-independent we have
\be
    0 = \dd\partial_a \dd v_b = \dd D_a \dd f_b + \dd h_{ab} \dd f + \dd\eta_b (\dd D_a \dd f - \dd f_a)\,, 
\ee
which then implies \(\dd f_a = \dd D_a \dd f\) and that \(\dd f\) satisfies \cref{eq:trans-cond}. Thus, any vector \(\dd v^a \in Ti^0\) gives rise to a Spi-translation in \(\mf t\). Conversely, given any translation \(\dd f \in \mf t\), the vector at \(i^0\) defined by (note the sign difference in the hyperboloidal component relative to \cref{eq:supetr-vec-field})
\be\label{eq:trans-vector-defn}
    \dd v^a \defn \dd f \dd \eta^a + \dd D^a \dd f\,,
\ee
is direction-independent i.e., \(\dd v^a \in Ti^0\). Thus, the Spi-translations \(\mf t\) can be represented by vectors in \(Ti^0\).
\end{remark}

\begin{remark}[Conformal transformation of Spi symmetries]\label{rem:conf-symm}
Let \((\dd f, \dd X^a)\) be a Spi symmetry defined by a vector field \(\xi^a\) as above, i.e.,
\be
    \dd X^a \defn \lim_{\to i^0} \Omega^{-\half} \xi^a \eqsp \dd f \defn \lim_{\to i^0} \Omega^{-1} \xi^a \eta_a\,.
\ee
For a fixed \(\xi^a\), consider the change in the conformal factor \(\Omega \mapsto \omega \Omega\). Then, from \cref{rem:conf} we have the transformations
\be
    \dd X^a \mapsto \dd X^a \eqsp \dd f \mapsto \dd f + \tfrac{1}{2} \lie_{\dd X}\dd \alpha\,.
\ee
Note that a pure supertranslation \((\dd f, \dd X^a = 0)\) is conformally-invariant, while a ``pure Lorentz'' symmetry \((\dd f = 0, \dd X^a)\) is not invariant but shifts by a supertranslation given by \(\tfrac{1}{2} \lie_{\dd X}\dd \alpha\). This further reflects the semi-direct structure of the \(\spi\) algebra given in \cref{eq:spi-semi-direct}.
\end{remark}

\hr 

To find the charge corresponding to the Spi-symmetries we need to evaluate the symplectic current \cref{eq:sympcurrent-H-simplified} when the perturbation denoted by \(\delta_2\) is generated by a Spi-symmetry. So we now calculate the perturbations $\delta_{(\dd f, \dd X)} \dd{E}$ and $\delta_{(\dd f, \dd X)} \dd{K}$ in the gravitational potentials corresponding to the metric perturbation \cref{eq:lin-diffeo}.

The potentials \(\dd E\) and \(\dd K_{ab}\) are defined in terms of (a rescaled) limit of \(S_{ab} \) by \cref{eq:potentials-defn}. Consider then the change in $S_{ab}$ under the perturbation \cref{eq:lin-diffeo}. The second term on the right-hand-side of \cref{eq:lin-diffeo} is a linearized conformal transformation (see \cref{rem:conf}) with \(\dd \alpha = - 2 \dd f\). Thus, the change in \(\dd E\) and \(\dd K_{ab}\) induced by this linearized conformal transformation is given by (see \cref{rem:conf-GR-fields})
\be \label{eq:supetr-deltaE-deltaK}
    \delta_{\dd{f}} \dd{E} = 0 \eqsp \delta _{\dd{f}} \dd{K}_{ab} = 4 (\dd D_{a}\dd D_{b} \dd{f} + \dd{f} \dd{h}_{ab})\,.
\ee

The first term on the right-hand-side of \cref{eq:lin-diffeo} is a linearized diffeomorphism and, since \(S_{ab}\) is a local and covariant functional of \(g_{ab}\) the corresponding perturbation in \(S_{ab}\) is \(\Lie_\xi S_{ab}\). Explicitly computing the Lie derivative, using \cref{eq:X-defn,eq:lorentz-cond} gives
\be
    \delta_{\dd X} \dd S_{ab} = \lim_{\to i^{0} } \Omh \lie_{\xi} S_{ab}=\dd X^{c} \dd{\partial}_{c} \dd{S}_{ab} +2 \dd{S}_{c(a}\dd{\eta}_{b)}\dd X^{c}+ 2 \dd{S}_{c(a}\dd{\partial}_{b)}\dd X^{c}\,.
\ee
Then, from the definition of the gravitational potentials \cref{eq:potentials-defn} we have
\be
    \delta_{\dd{X}} \dd{E} = \lie_{\dd X} \dd{E} \eqsp \delta_{\dd{X}} \dd{K}_{ab} = \lie_{\dd{X}} \dd{K}_{ab}\,.
\ee
As a result, under a general Spi symmetry parametrized by $(\dd{f}, \dd{X}^{a})$ we have
\be \label{eq:spi-changes}
    \delta_{(\dd f, \dd X)} \dd{E} = \lie_{\dd X}\dd E \eqsp \delta_{(\dd f, \dd X)} \dd{K}_{ab} = \lie_{\dd X} \dd{K}_{ab} + 4 (\dd D_{a}\dd D_{b} \dd{f} + \dd{h}_{ab} \dd{f})\,.
\ee
Note that our parity condition \cref{eq:parity-E} does not place any further restrictions on these symmetries.\\

\begin{remark}[Special choices of conformal factor]\label{rem:GR-gauge-choice}
The freedom in the conformal factor can be used to impose further restrictions on the potential \(\dd K_{ab}\). We note the following two conditions that have been used in prior work.
\begin{enumerate}
    \item From \cref{eq:conf-K} we see that \(\dd K \defn \dd h^{ab}\dd K_{ab}\) transforms as
    \be
        \dd K \mapsto \dd K - 2 (\dd D^2 \dd\alpha + 3 \dd\alpha)\,.
    \ee
    Now given a choice of conformal factor so that \(\dd K \neq 0\) we can always solve a linear hyperbolic equation for \(\dd\alpha\) on \(\hyp\) and choose a new conformal factor (as in \cref{rem:conf}) so that in the new conformal completion \(\dd K = 0\). This is the choice made in \cite{CD,CDV-Lorentz,Tro}. With this restriction on \(\dd K\) we see from \cref{eq:spi-changes} that the allowed supertranslations are reduced to functions \(\dd f\) which satisfy
    \be\label{eq:st-CD}
        \dd D^2 \dd f + 3 \dd f = 0\,.
    \ee
    \item Consider the restricted class of spacetimes where \(\dd B_{ab} = 0\). Then, the tensor \(\dd K_{ab}\) can be written in terms of a scalar potential \(\dd k\) as in \cref{eq:K-potential}. Comparing \cref{eq:K-potential} with \cref{eq:conf-K} we see that we can choose \(\dd \alpha = \half \dd k\). Then, we can choose a new conformal factor (as in \cref{rem:conf}) so that in the new conformal completion \(\dd K_{ab} = 0\). This is the choice made in \cite{AH,Ash-in-Held}. With this restriction we see from \cref{eq:spi-changes} that the allowed supertranslations are reduced to the translation algebra (\cref{eq:trans-cond}), and the full asymptotic symmetry algebra reduces to the Poincar\'e algebra.
\end{enumerate}
\end{remark}
It is not clear, a priori, what such special choices of conformal factor imply at null infinity. From the point of view of matching the Spi symmetries and charges to the ones defined on null infinity such choices of conformal factors might not be convenient. So we will \emph{not} impose any such conditions on the conformal factor in our analysis and work with the full \(\spi\) algebra. However, we will argue that our results reduce to those of \cite{AH,CD} when the corresponding restrictions are imposed.

%%==========================================================
\section{Spi-charges}
\label{sec:spi-charges}
In this section we now compute the charges associated with the Spi-symmetries. Following our strategy we consider the symplectic current \(\pb{\dd\omega}\) where one of the perturbations, \(\delta_2\), is a perturbation generated by an asymptotic Spi-symmetry represented by \((\dd f, \dd X^a)\). Using \cref{eq:sympcurrent-H-simplified,eq:spi-changes} we have
\be\label{eq:omega-symm}
    \pb{\dd\omega}(\delta g, \delta_{(\dd f, \dd X)} g) = \frac{1} {64\pi} \dd\varepsilon_3 \lb[ \delta\dd K \lie_{\dd X} \dd E - \delta \dd E \lie_{\dd X} \dd K -  4 \delta \dd E (\dd D^2 \dd f + 3 \dd f) \rb]\,. 
\ee
We show next that, under suitable conditions, the above expression can be written as a total derivative on \(\hyp\) that is, 
\be\label{eq:omega-Q}
    \pb{\dd\omega}(\delta g, \delta_{(\dd f, \dd X)} g) = - \dd\varepsilon_3~ \dd D^a \dd Q_a(g; \delta g; (\dd f, \dd X))\,,
\ee
where \(\dd Q_a\) is a local and covariant functional of its arguments on \(\hyp\).

It will be convenient to do this separately for supertranslations and Lorentz symmetries. In \cref{sec:st}, we will find that for supertranslations the functional \(\dd Q_a\) is integrable, and defines the \emph{supermomentum} charges on cross-sections \(S\) of \(\hyp\). Then we show in \cref{sec:lorentz-charge} that for Lorentz symmetries \(\dd Q_a\) is not integrable, in general. In this case we will adopt the prescription of Wald and Zoupas with suitable modifications to define an integrable charge for Lorentz symmetries. Finally, as noted in \cref{rem:conf-symm}, a ``pure Lorentz'' symmetry is not conformally-invariant but shifts by a supertranslation. Similarly, we show in \cref{sec:conf-charges} that the Lorentz charge shifts by a supertranslation charge under conformal transformations, in accord with the semi-direct structure of the \(\spi\) algebra (\cref{eq:spi-semi-direct}).

%%----------------------------------------------------------
\subsection{Charges for supertranslations: Spi-supermomentum}
\label{sec:st}

To define the charge for the supertranslations consider \cref{eq:omega-symm} for a pure supertranslation \((\dd f, \dd X^a = 0)\)
\be\begin{aligned}
    \pb{\dd\omega}(\delta g, \delta_{\dd f} g) & = - \frac{1}{16\pi} \dd\varepsilon_3~ \delta \dd E (\dd D^2 \dd f + 3 \dd f)\,, \\
    & = - \frac{1}{16\pi} \dd\varepsilon_3 \dd D^{a} \delta (\dd{E} \dd D_{a} \dd{f} - \dd{f}\dd D_{a}\dd{E})\,,
\end{aligned}\ee
where the second line uses \cref{eq:box-E}. In this case, the symplectic current can be written in the form \cref{eq:omega-Q} where the \(\dd Q_a\) is manifestly integrable. Thus, we define the Spi \emph{supermomentum} charge at a cross-section \(S\) of \(\hyp\) by 	 
\be \label{eq:st-charge}
    \mathcal{Q}[\dd{f}; S] = \frac{1}{16\pi} \int_S \dd\varepsilon_2~  \dd u^a (\dd{E} \dd D_{a} \dd{f} - \dd{f}\dd D_{a}\dd{E})\,.
\ee
Here we have chosen the charge to vanish on Minkowski spacetime where \(\dd E = 0\) (see \cref{eq:Mink-stuff}). The corresponding flux is given by (using \cref{eq:box-E}) 
\be\label{eq:st-flux}
    \mathcal{F}[\dd{f};\Delta \hyp] \defn \mathcal{Q}[\dd{f}; S_2] - \mathcal{Q}[\dd{f}; S_1] = - \frac{1}{16\pi}\int_{\Delta \hyp} \dd{\varepsilon}_{3}~ \dd{E} (\dd D^{2} \dd{f} + 3\dd{f})\,.
\ee

When \(\dd f \in \mf t\) is a Spi-translation the charge \cref{eq:st-charge} can be written in an alternative form as follows: Using \cref{eq:EB-potentials,eq:box-E} we have the identity
\be\label{eq:st-tr-conversion}\begin{aligned}
    -\dd{f} \dd D_{a} \dd{E} + \dd{E} \dd D_{a} \dd{f} & = 2 \dd{E}_{ab} \dd D^{b} \dd{f} + \dd D^{b} \lb( \dd D_{[a} \dd{E} \dd D_{b]} \dd{f} \rb) \\
    &\quad - \tfrac{1}{2} \lb[ \dd D_a \dd E (\dd D^2 \dd f + 3 \dd f ) - \dd D^b \dd E (\dd D_a \dd D_b \dd f + \dd h_{ab} \dd f) \rb] \,.
\end{aligned}\ee
The second term on the right-hand-side corresponds to an exact \(2\)-form and vanishes upon integrating on \(S\), while the last line vanishes for translations due to \cref{eq:trans-cond}. Hence, the charge for any translation \(\dd f \in \mf t\) can be written as
\be \label{eq:tr-charge}
    \mc Q[\dd{f}; S] = \frac{1}{8\pi} \int_S \dd\varepsilon_2~ \dd u^{a} \dd{E}_{ab} \dd D^{b} \dd{f}\,,
\ee
which reproduces the charge for translations given in \cite{AH}. Using \cref{eq:trans-cond} the flux of translations vanishes across any region \(\Delta\hyp\) and thus the translation charge is independent of the choice of cross-section \(S\). Using the isomorphism between Spi-translations \(\dd f\) and vectors \(\dd v^a\) in \(Ti^0\) (see \cref{rem:trans-vectors}), the translation charge in \cref{eq:tr-charge} defines a \(4\)-momentum vector \(\dd P^a\) at \(i^0\) such that
\be
    \dd P^a \dd v_a = \mc Q[\dd{f}; S]\,.
\ee
Note that this relation is well-defined at \(i^0\) since the translation charge is independent of the cross-section \(S\). The vector \(\dd P^a\) is precisely the ADM \(4\)-momentum at \(i^0\) \cite{AMA-spi-3+1} and also coincides with the limit to \(i^0\) of the Bondi \(4\)-momentum on null infinity \cite{Ash-Mag-Ash} (the corresponding result for all the supertranslation charges was proven in \cite{KP-GR-match}).

The charge expression \cref{eq:st-charge} agrees with the results of Comp\`ere  and Dehouck \cite{CD}. Note that when the conformal factor is chosen so that \(\dd K = 0\) the supertranslation algebra is reduced to the subalgebra satisfying \cref{eq:st-CD} and the flux corresponding to such supertranslations vanishes across any region \(\Delta\hyp\). As was shown in \cite{KP-GR-match}, to relate the supertranslation symmetries and charges at spatial infinity to the ones on null infinity, it is sufficient that the \emph{total} flux of these charges vanishes on \emph{all} of \(\hyp\),\footnote{To make this rigorous it is necessary to additionally complete \(\hyp\) to include the null directions at \(i^0\). This construction is detailed in \cite{KP-EM-match,KP-GR-match}.} and the flux need not vanish across some local region \(\Delta\hyp\). Thus the restriction on the conformal factor imposing \(\dd K = 0\) is not necessary.

Note that in \cite{KP-GR-match} the supermomentum charges at spatial infinity were related to those on null infinity using the Ashtekar-Hansen expression \cref{eq:tr-charge} for \emph{all} supertranslations (even those which are not translations), instead of the expression \cref{eq:st-charge}. On \(\hyp\), these charge expressions differ by the integral of last line of \cref{eq:st-tr-conversion} over some cross-section \(S\). However, the regularity conditions on \(\dd E\) and \(\dd f\) used in \cite{KP-GR-match} as the spatial directions \(\vec\eta\) limit to null directions at \(i^0\) ensure that the additional terms vanish (see, for instance, Appendix.~D of \cite{KP-GR-match}) and both expressions yield the same \emph{finite} supermomenta in null directions which further equals the supermomenta at null infinity. Thus, the result of \cite{KP-GR-match} can also be derived using the expression \cref{eq:st-charge} for the supertranslation charges.

%%---------------------------------------------------------
\subsection{Lorentz charges with \(\dd B_{ab} = 0\)}
\label{sec:lorentz-charge}

Next we will obtain a charge formula for the Lorentz symmetries. As emphasized in \cite{AH,Ash-in-Held}, to obtain such a charge formula one needs to consider the ``subleading'' piece of the magnetic part of the Weyl tensor. Thus, in the following we will make the additional assumption that \(\dd B_{ab} = 0\) and that the ``subleading'' magnetic part \(\dd \beta_{ab}\) defined in \cref{eq:beta-defn} exists. However, in \cref{sec:new-beta} we show how the restriction that \(\dd B_{ab}\) vanishes can be lifted to obtain a charge for the Lorentz symmetries.

For a ``pure Lorentz'' symmetry \((\dd f = 0, \dd X^a)\) we have from \cref{eq:omega-symm}
\be \label{symp-lorentz}
    \pb{\dd\omega}(\delta g, \delta_{\dd X}g) = \frac{1}{64 \pi} \dd{\varepsilon}_{3} ( \lie_{\dd X} \dd{E} \delta \dd{K} - \lie_{\dd X} \dd{K} \delta \dd{E})\,.
\ee
We now want to write this as a total derivative of the form \cref{eq:omega-Q}. To do so consider the following tensor
\be\label{eq:W-defn}
    \dd{W}_{ab} \defn \dd{\beta}_{ab} + \tfrac{1}{8} \dd{\varepsilon}_{cd(a} \dd D^{c} \dd{E} \dd K^d{}_{b)} - \tfrac{1}{16} \dd{\varepsilon}_{abc}\dd{K}\dd D^{c}\dd{E}\,. 
\ee
Using \cref{eq:div-beta,eq:vanishing-curl-K,eq:div-K}, we obtain
\be
    \dd D^a \dd W_{ab} = 0 \eqsp \dd h^{ab} \dd W_{ab} = 0\,.
\ee
Note that \(\dd W_{ab}\) is not a symmetric tensor. Further using \cref{eq:W-defn,eq:box-X} we have
\be \label{div-eq}
    \dd D^{a}[\dd{W_{ab}}\dv{\dd{X}}^{b}] = \tfrac{1}{8} \dd{X}^{a} \dd D_{a}\dd{E} \dd{K}\,,
\ee
where $\dv{\dd{X}}^{a} \defn \f{1}{2} \dd{\varepsilon}^{abc}\dd D_{b} \dd{X}_{c}$ is the ``dual'' Killing vector field to \(\dd X^a\) (see \cref{eq:dual-X-defn}). Therefore, \cref{symp-lorentz} can be written as
\be\label{eq:symp-lorentz-W}
    \pb{\dd\omega}(\delta g, \delta_{\dd X}g) = \frac{1}{8 \pi} \dd{\varepsilon}_3 \dd D^a \lb[  \delta \dd W_{ab} \dv X^b - \tfrac{1}{8} \delta \dd E \dd K \dd X_{a} \rb]\,,
\ee
which is again of the form \cref{eq:omega-Q}. However the functional \(\dd Q_a\) in this case is not integrable, in general. To see this consider
\be\label{eq:Q-lor}
    \int_S \dd\varepsilon_2~ \dd u^a \dd Q_a [\delta g ; \dd X] = - \frac{1}{8 \pi} \int_S \dd\varepsilon_2~ \dd u^a \lb[  \delta \dd W_{ab} \dv X^b - \tfrac{1}{8} \delta \dd E \dd K \dd X_{a} \rb]\,,
\ee
and compute an antisymmetrized second variation to get
\be\label{eq:integrability}\begin{aligned}
    \int_S \dd\varepsilon_2 \dd u^a \big(\delta_1 \dd Q_a [\delta_2 g ; \dd X] - \delta_2 \dd Q_a [\delta_1 g ; \dd X] \big) & = \tfrac{1}{64\pi} \int_S \dd\varepsilon_2 \dd u^a \dd X_a \lb(\delta_1 \dd K \delta_2 \dd E  - \delta_2 \dd K \delta_1 \dd E   \rb) \\
    & = - \int_S \dd X \cdot \pb{\dd\omega}(\delta_1 g, \delta_2 g)\,.
\end{aligned}\ee
If \cref{eq:Q-lor} were integrable then the above antisymmetrized second variation would vanish for all perturbations and all cross-sections \(S\). However, since we allow arbitrary perturbations of both \(\dd E\) and \(\dd K_{ab}\) the expression on the right-hand-side vanishes if and only if the Lorentz vector field happens to be tangent to the cross-section \(S\). However a general Lorentz vector field is not tangent to any cross-section of \(\hyp\), in particular Lorentz boosts do not preserve any cross-section of \(\hyp\). Thus, the expression \cref{eq:Q-lor} is not integrable and cannot be used to define the charge of Lorentz symmetries.

To remedy this, note that \cref{eq:integrability} is similar to the integrability criterion derived by Wald and Zoupas (see Eq.~16 of \cite{WZ}). Wald and Zoupas further developed a general prescription to define a integrable charge (``conserved quantity'') which we now adapt to our case. Let \(\dd\Theta(g;\delta g)\) be a \(3\)-form on \(\hyp\) which is a symplectic potential for the pullback of the symplectic current (\cref{eq:sympcurrent-H-simplified}) to $\hyp$, that is,
\be\label{eq:Theta-symppot}
    \pb{\dd\omega}(g; \delta_{1} g, \delta_{2} g) =\delta_{1} \dd{\Theta}(g; \delta_{2}g) - \delta_{2} \dd{\Theta}(g; \delta_{1}g)\,, 
\ee
for all backgrounds and all perturbations. We also require that the choice of \(\dd\Theta\) satisfy the following conditions
\begin{enumerate}
    \item \(\dd\Theta\) is locally and covariantly constructed out of the dynamical fields \((\dd E, \dd K_{ab})\), their perturbations, and finitely many of their derivatives, along with the ``universal background structure'' \(\dd h_{ab}\) present on \(\hyp\).
    \item \(\dd\Theta\) is independent of any arbitrary choices made in specifying the background structure, in particular, \(\dd\Theta\) is conformally-invariant.
    \item  $\dd\Theta(g;\delta g) = 0$ for Minkowski spacetime for \emph{all} perturbations $\delta g$.
\end{enumerate}

In analogy to the Wald-Zoupas prescription we define the charge \(\mc Q[\dd X^a; S]\) associated with a Lorentz symmetry through
\be\label{eq:WZ-charge}
    \delta \mc Q[\dd X^a ; S] \defn \int_S \dd\varepsilon_2 \dd u^a \dd Q_a(\delta g; \dd X^a) + \int_S \dd X \cdot \dd\Theta(\delta g)\,. 
\ee
From \cref{eq:integrability,eq:Theta-symppot} it follows that the above defining relation is integrable and thus defines a charge \(\mc Q[\dd X^a ; S]\) once we pick a reference solution where the charge vanishes.\\

For the \(3\)-form \(\dd\Theta\) we choose
\be \label{eq:Theta-choice}
    \dd{\Theta}(g;\delta g) \defn -\frac{1}{64\pi}\dd{\varepsilon}_{3} \dd{E} \delta \dd{K} \,.
\ee
It can be verified that this choice satisfies all the criteria listed below \cref{eq:Theta-symppot}. In particular \(\dd\Theta\) is conformally-invariant, and for Minkowski spacetime \(\dd E = 0\) (\cref{eq:Mink-stuff}) and so \(\dd\Theta = 0\) on Minkowski spacetime for \emph{all} perturbations. This choice for \(\dd\Theta\) is not unique, but we will argue in \cref{sec:amb} that the ambiguity in the the choice of \(\dd\Theta\) does not affect our final charge expression.

With the choice \cref{eq:Theta-choice} and \cref{eq:Q-lor,eq:WZ-charge}, we have
\be
    \delta \mc Q[\dd{X}^{a};S] = - \frac{1}{8 \pi} \int_S\dd{\varepsilon}_{2}~ \dd u^{a} \delta [ \dd{W}_{ab} \dv{\dd{X}}^{b} - \tfrac{1}{8} \dd{K} \dd{E} \dd{X}_{a}]\,,
\ee
We define the unperturbed charge by picking the reference solution to be Minkowski spacetime which satisfies \(\dd E = 0\) and \(\dd \beta_{ab} = 0\) (\cref{eq:Mink-stuff}). Thus, we have the charge
\be\label{eq:lorentz-charge}
    \mc Q[\dd{X}^{a};S] = - \frac{1}{8 \pi} \int_S\dd{\varepsilon}_{2}~ \dd u^{a} [ \dd{W}_{ab} \dv{\dd{X}}^{b} - \tfrac{1}{8} \dd{K} \dd{E} \dd{X}_{a}]\,,
\ee
The corresponding flux of the Lorentz charges is given by
\be\label{eq:lorentz-flux}
    \mathcal{F}[\dd{X}^{a}, \Delta \hyp] = - \frac{1}{64\pi}\int_{\Delta \hyp} \dd\varepsilon_3~ \dd{E} \lie_{\dd X}\dd{K}\,. 
\ee
Note that the flux is essentially given by \(\mathcal{F}[\dd{X}^{a}, \Delta \hyp] = \int_{\Delta\hyp} \dd\Theta(g;\delta_{\dd X}g)\) in analogy to the Wald-Zoupas prescription (see Eq.~32 of \cite{WZ}).\\

When the conformal factor is chosen so that \(\dd K_{ab} = 0\) then the Lorentz charge reduces to
\be\label{eq:lorentz-charge-AH}
    \mc Q [\dd X^a; S] = - \f{1}{8 \pi } \int_S \dd\varepsilon_{2}~ \dd u^{a} \dd\beta_{ab} \dv X^b\,,
\ee
which is the expression given by \cite{AH}. Note that when the conformal factor is chosen such that \(\dd K = 0\), the expression \cref{eq:Q-lor} is manifestly integrable and our ``correction term'' \(\dd\Theta\) (\cref{eq:Theta-choice}) vanishes. In both these cases, the flux of the Lorentz charges vanishes across any region \(\Delta\hyp\), i.e., the Lorentz charges are identically conserved. Further, since the vector fields \(\dd X^a\) correspond precisely to infinitesimal Lorentz transformations \(\dd\Lambda_{ab}\) in \(Ti^0\) (see \cref{eq:X-Lambda}), the charge defines an ``angular momentum'' tensor \(\dd J^{ab}\) at \(i^0\) through
\be
    \dd J^{ab} \dd\Lambda_{ab} = \mc Q [\dd X^a; S]\,,
\ee
where the right-hand-side is independent of the cross-section since the charge is conserved.

%%--------------------------------------------------------
\subsection{Transformation of charges under conformal changes}
\label{sec:conf-charges}

We now consider the transformation of the charges and fluxes for a Spi symmetry under changes of the choice of conformal factor as discussed in \cref{rem:conf}.

Consider a pure supertranslation symmetry \((\dd f, \dd X^a = 0)\). As shown in \cref{rem:conf-symm}, a pure supertranslation is conformally-invariant. Further from \cref{rem:conf-GR-fields} the potential \(\dd E\) is also conformally-invariant. Thus, the charge and flux of supertranslations in \cref{eq:st-charge,eq:st-flux} is also conformally-invariant.

However a ``pure Lorentz'' symmetry \((\dd f = 0, \dd X^a)\) is not conformally-invariant (see \cref{rem:conf-symm}), and hence we expect that the charge and flux of a Lorentz symmetry must transform nontrivially under changes of the conformal factor. Consider first the flux of Lorentz charges given by \cref{eq:lorentz-flux}. Using the transformation of \(\dd K_{ab}\) (\cref{eq:conf-K}) we see that this flux expression transforms as
\be
    \mc F[\dd X^a; \Delta \hyp ] \mapsto \mc F[\dd X^a; \Delta \hyp ] + \frac{1}{32\pi} \int_{\Delta\hyp} \dd\varepsilon_3 \dd E (\dd D^2 \Lie_{\dd X} \dd\alpha + 3 \Lie_{\dd X} \dd\alpha)\,.
\ee
Comparing the second term on the right-hand-side to \cref{eq:st-flux}, we see that it is precisely the flux of a supertranslation given by \((-\half \Lie_{\dd X} \dd\alpha)\). Thus, under a change of conformal factor the Lorentz flux shifts by the flux of a supertranslation
\be\label{eq:conf-lorentz-flux}
    \mc F[\dd X^a; \Delta \hyp ] \mapsto \mc F[\dd X^a; \Delta \hyp ] + \mc F[-\half \Lie_{\dd X}\dd\alpha; \Delta\hyp]\,.
\ee

One can similarly verify that the Lorentz charge \cref{eq:lorentz-charge} also shifts by the charge of a supertranslation. The explicit computation is a bit tedious and is presented in \cref{sec:conf-lorentz-charge}. However, we can derive the transformation of the Lorentz charge by a more general argument which we present below. This argument also holds in the more general case when \(\dd B_{ab} \neq 0\) considered in \cref{sec:new-beta} below.

From the transformation of the flux \cref{eq:conf-lorentz-flux}, we can deduce that the Lorentz charge expression \cref{eq:lorentz-charge} must transform as
\be
    \mc{Q}[\dd X^a; S] \mapsto \mc Q[\dd X^a; S] + \mc Q[-\nfrac{1}{2}\lie_{\dd X}\dd \alpha; S] + \int_S \dd\varepsilon_2 \dd u^a  \dd \mu_a[\dd\alpha]\,,
\ee
where the second term on the right-hand-side is the charge of a supertranslation \((-\half \Lie_{\dd X}\dd\alpha)\) and the third term is a possible additional term determined by a covector \(\dd\mu_a\) which depends linearly on \(\dd\alpha\) and is divergence-free, \(\dd D^a \dd\mu_a[\dd\alpha] = 0\) for \emph{all} \(\dd\alpha\). Since \(\dd\alpha\) is a free function on \(\hyp\) we can apply \cref{thm:Wald} with \(\dd\alpha\) as the ``dynamical field''. Thus, from \cref{eq:Wald-hyp} we conclude that the final integral above vanishes, and that the Lorentz charge shifts by the charge of a supertranslation \((-\half \Lie_{\dd X}\dd\alpha)\).
\be\label{eq:lor-charge-trans-summ}
   \mc{Q}[\dd X^a; S] \mapsto \mc Q[\dd X^a; S] + \mc Q[-\nfrac{1}{2}\lie_{\dd X}\dd \alpha; S]\,.
\ee

If we restrict to the choice of conformal factor where \(\dd K_{ab} = 0\), so that the asymptotic symmetries are reduced to the Poincar\'e algebra and \(\dd\alpha\) is a Spi-translation satisfying \cref{eq:trans-cond}, then \cref{eq:lor-charge-trans-summ} reproduces the transformation law given in Eq.~29 of \cite{AH} and Eq.~6.8 of \cite{Ash-in-Held}.

Consider the charge of any Spi-symmetry represented by \((\dd f, \dd X^a)\), then under a conformal transformation the same Spi-symmetry is now represented by \((\dd f + \nfrac{1}{2} \lie_{\dd X}\dd\alpha, \dd X^a)\) (see \cref{rem:conf-symm}). The total charge of the Spi-symmetry transforms as
\be\label{eq:total-charge-inv}\begin{aligned}
    \mc{Q}[\dd f; S] + \mc{Q}[\dd X^a; S] \mapsto & \mc{Q}[\dd f + \nfrac{1}{2} \lie_{\dd X}\dd\alpha; S] + \mc{Q}[\dd X^a; S] + \mc Q[-\nfrac{1}{2}\lie_{\dd X}\dd \alpha; S]\,, \\
    & = \mc{Q}[\dd f; S] + \mc{Q}[\dd X^a; S]\,,
\end{aligned}\ee
that is, the charge of any Spi-symmetry is independent of the choice of conformal factor --- the change in the function \(\dd f\) representing the symmetry is exactly compensated by the change in the Lorentz charge given in \cref{eq:lor-charge-trans-summ}.

%%=========================================================
\section{Discussion}
\label{sec:disc}
In this paper, we analyzed the asymptotic symmetries and the corresponding charges for asymptotically-flat spacetimes at spatial infinity \(i^0\) using the Ashtekar-Hansen formalism, without any restrictions on the choice of the conformal factor at spatial infinity, which were imposed in previous analyses. Using the covariant phase space, we considered the direction-dependent limit of symplectic current of vacuum general relativity to spatial infinity. Using the pullback of this limit of the symplectic current to the space of spatial directions \(\hyp\) at spatial infinity, we obtained expressions for charges corresponding to all asymptotic symmetries. We rederived the known expressions for supertranslation charges but more a general expression for the Lorentz charge when conformal factor is completely unrestricted. In this case, we used a Wald-Zoupas type correction to make the Lorentz charge integrable, which also ensures that this charge transforms correctly under the action of a supertranslation, or equivalently, that the charge of a general Spi-symmetry is conformally-invariant.

The main motivation behind our analysis is to eventually relate the Lorentz charges at spatial infinity to the ones defined on null infinity. In this context, the Lorentz charge expressions would have to be matched in the ``same'' choice of conformal factor at both null infinity and spatial infinity, and it is not clear what the restrictions on the conformal factor at spatial infinity placed in previous works imply at null infinity. Thus, we hope that our more general expression for the Lorentz charge at spatial infinity will be more useful to repeat the matching analysis for the case of Lorentz symmetries that was done previously for Maxwell theory \cite{KP-EM-match} and supertranslations in general relativity \cite{KP-GR-match}. If this works out as expected, this would imply that the full BMS group at past null infinity is matched to the full BMS group at future null infinity and moreover, that the incoming fluxes of all BMS symmetries through past null infinity are equal to the outgoing fluxes of the anitpodally identified BMS symmetries through future null infinity. This would then prove the existence of infinitely many conservation laws, one for each generator of the BMS group, in classical gravitational scattering in asymptotically-flat spacetimes, as anticipated by Strominger \cite{Stro-CK-match}.

Another avenue for future investigation would be to quantize the asymptotic fields on \(\hyp\) in the spirit of the asymptotic quantization program on null infinity \cite{Ashtekar:1987tt}, see also \cite{Alexander1984}. This could lead to the possibility of relating the asymptotic ``in-states'' on past null infinity to the ``out-states'' on future null infinity, similar to the matching conditions in the classical theory, and provide further insight into the structure of quantum scattering.

We also note that the asymptotic fields at spatial infinity in both Maxwell theory and general relativity are described by smooth tensor fields living on a unit-hyperboloid $\hyp$. As is well-known \(\hyp\) is precisely the \(3\)-dimensional de Sitter spacetime. To prove the matching conditions for Maxwell and gravitational fields on \(\hyp\) with those on null infinity, \(\hyp\) was conformally-completed into a cylinder in the analysis of \cite{KP-GR-match,KP-EM-match}. It would be interesting to see if insights from the de Sitter/CFT correspondence \cite{dS-CFT} can be applied to develop a holographic understanding of electromagnetism and general relativity in asymptotically-flat spacetimes at spatial infinity, perhaps similar to \cite{Mink-CFT}.

%%======================================================
\section*{Acknowledgements}

We thank \'Eanna \'E. Flanagan for helpful discussions and constant encouragement over the course of this work. IS would also like to thank D. Iozzo for help with \textsc{xAct}. This work is supported in part by the NSF grant PHY-1707800 to Cornell University. Some calculations used the computer algebra system \textsc{Mathematica} \cite{Mathematica}, in combination with the \textsc{xAct/xTensor} suite~\cite{JMM:xAct,MARTINGARCIA2008597}, and the Riemannian Geometry and Tensor Calculus package \cite{RGTC}.

\appendix

%%===========================================================
\section{Coordinates, universal structure and asymptotic expansions near \(i^0\)}
\label{sec:coord}

In this appendix we construct a suitable asymptotic coordinate system near spatial infinity. Using these coordinates we explicitly demonstrate the universal structure near \(i^0\) described in \cref{sec:univ-str}. We also describe the asymptotic expansion of the unphysical and physical metrics in these coordinates, thus making contact with the expansions used in previous works \cite{Beig-Schmidt,CD,CDV-Lorentz}.

Consider the unphysical spacetime \((M,g_{ab})\) obtained from some physical spacetime satisfying \cref{def:AH}. The unphysical metric \(\dd g_{ab}\) at \(i^0\) induces a metric which is isometric to the Minkowski metric in the tangent space \(Ti^0\). Thus we can introduce asymptotically Cartesian coordinates \((t,x,y,z)\) so that \(i^0\) is at the origin of this coordinate system and
\be
    \dd g_{ab} \equiv -dt^2 + dx^2 + dy^2 + dz^2\,.
\ee
Note that \(x^i = (t,x,y,z)\) define a \(C^1\) coordinate system at \(i^0\). To define a \(C^{>1}\) differential structure we allow any other coordinate chart \({x'}^i(x)\) such that
\be\label{eq:C>1-struc}
    \frac{\partial^2 {x'}^i(x)}{\partial x^j \partial x^k} \text{ and } \frac{\partial^2 x^i(x')}{\partial {x'}^j \partial {x'}^k} \text{ are } C^{>-1} \text{ at } i^0\,.
\ee
A collection of all coordinate charts related by \cref{eq:C>1-struc} defines a choice of \(C^{>1}\)-structure on \(M\) at \(i^0\), see \cite{Herb-dd} and Appendix~A of \cite{KP-EM-match} for details.

It is more convenient to use coordinates which are adapted to the space of unit spacelike directions \(\hyp\). Thus define \((\rho, \tau)\) by
\be
    \rho^2 \defn - t^2 + x^2 + y^2 + z^2 \eqsp \tanh \tau \defn \frac{t}{\sqrt{x^2 + y^2 + z^2}}\,.
\ee
In these coordinates the metric in \(Ti^0\) takes the form
\be
    \dd g_{ab} \equiv d\rho^2 + \rho^2 \lb( - d\tau^2 + \cosh^2\tau s_{AB} d\theta^A d\theta^B \rb)\,,
\ee
where \(s_{AB}\) is the unit metric on \(\bb S^2\) in some coordinates \(\theta^A\), say the usual \((\theta,\phi)\) coordinates. Note that the coordinates \((\rho,\tau,\theta^A)\) are \emph{not} \(C^{>1}\) coordinates --- the bases \((d\rho, \rho d\tau, \rho d\theta^A)\) are not continuous but are direction-dependent at \(i^0\).

 The unit spatial directions \(\vec\eta\) then correspond to the unit vectors \(\partial_\rho\) in \(Ti^0\) which are parameterized by \((\tau, \theta^A)\). The space of directions \(\hyp\) is then the surface \(\rho = 1\) in \(Ti^0\) with the induced metric
\be
    \dd h_{ab} \equiv - d\tau^2 + \cosh^2\tau s_{AB} d\theta^A d\theta^B\,.
\ee
The reflection of the directions \(\vec\eta \mapsto - \vec\eta\) then induces the reflection isometry
\be\label{eq:reflection}
    (\tau,\theta^A) \mapsto (-\tau, -\theta^A)\,,
\ee
on \(\hyp\), where \(\theta^A \mapsto -\theta^A\) is the antipodal reflection on \(\bb S^2\).

So far we have only considered the structure at \(i^0\) , now we extend the metric away from \(i^0\). Since the unphysical metric \(g_{ab}\) is \(C^{>0}\) and limits to \(\dd g_{ab}\) at \(i^0\) (where \(\rho = 0\)), it can be verified that \(g_{ab}\) admits an expansion in \(\rho\) of the form 
\be\label{eq:g-exp}\begin{aligned}
    g_{ab} & \equiv \lb[ 1 + \sigma \rho + o(\rho) \rb]^2 d\rho^2 + 2 \lb[ \rho A_a + o\left(\rho\right) \rb] d \rho (\rho d y^{a}) \\
    &\quad + \lb[ h_{ab}^{(0)} + \rho h_{ab}^{(1)}
 + o(\rho) \rb] (\rho dy^{a}) (\rho dy^{b})\,,
\end{aligned}\ee
where \(y^a = (\tau, \theta^A)\) are coordinates on the unit hyperboloid, and \(h_{ab}^{(0)} \equiv \dd h_{ab}\) is the unit hyperboloid metric. The expansion coefficients \(\sigma\), \(A_a\) and \(h_{ab}^{(1)}\) can be considered as tensor fields on \(\hyp\). The \(o(\rho)\) denotes terms which falloff faster than \(\rho\) in the limit to \(i^0\), that is, \(\lim_{\rho \to 0} \rho^{-1} o(\rho) = 0\).

For the conformal factor, one can choose
\be\label{eq:Omega-BS}
    \Omega = \rho^2\,,
\ee
which can be verified to satisfy all the conditions in \cref{def:AH}, that is, in the limit \(\rho \to 0\), \(\Omega = 0\), \(\nabla_a \Omega = 0 \) and \(\nabla_a \nabla_b \Omega = 2 g_{ab}\). Before considering the physical metric lets analyze the universal structure at \(i^0\).

From the above discussion it is clear that the metric \(\dd g_{ab}\) and the space of directions \(\hyp\) is universal, that is, independent of which unphysical metric is chosen. What is the structure corresponding to the equivalence classes of \(C^{>1}\) curves described in \cref{sec:univ-str}? Consider the \(C^{>1}\) curves \(\Gamma_v\) through \(i^0\) with tangents \(v^a \equiv \partial_\rho\) in these coordinates. Further, with the choice of conformal factor in \cref{eq:Omega-BS} we have
\be
    \eta^a = \nabla^a \Omh \equiv (1 - 2 \rho \sigma ) \frac{\partial}{\partial \rho} + \rho h^{(0)ab} A_b \frac{\partial}{\rho\partial y^a} + o(\rho)\,.
\ee
From \cref{eq:equiv-curves} we see that the curves \(\Gamma_v\) (with tangent \(v^a \equiv \partial_\rho\)) will be equivalent to the curves \(\Gamma_\eta\) (with tangent \(\eta^a\)) for all spacetimes if we can always choose \(A_a\) to vanish. This can be accomplished using the freedom in the choice of the hyperboloid coordinates \(y^a\) at ``next order'' in \(\rho\). Consider the coordinate transformation\footnote{This is essentially the unphysical spacetime version of the coordinate transformations consider in Lemma~2.2 of \cite{Beig-Schmidt}.}
\be\label{eq:BS-trans}
    \rho \mapsto \rho \eqsp y^a \mapsto y^a + \rho h^{(0)ab} A_b\,.
\ee
By rewriting this in terms of the Cartesian coordinates \(x^i = (t,x,y,z)\), it can be verified that the transformation \cref{eq:BS-trans} is a \(C^{>1}\) coordinate transformation (\cref{eq:C>1-struc}). It can be also be verified that using this transformation the \(d\rho dy^a\) term in the metric, i.e. \(A_a\), vanishes in the new coordinates. Thus, the curves \(\Gamma_v\) and \(\Gamma_\eta\) can always be chosen to be equivalent. Further, this choice can always be made in any choice of the physical spacetime. Thus, the equivalence classes of \(C^{>1}\) curves through \(i^0\) is also universal.

Having made this choice the unphysical metric takes the form
\be\begin{aligned}
    g_{ab} & \equiv \lb[ 1 + \sigma \rho + o(\rho) \rb]^2 d\rho^2 + \rho\, o\left(\rho\right) d \rho d y^{a} +\rho^{2}\lb[ h_{ab}^{(0)} + \rho h_{ab}^{(1)}
 + o(\rho) \rb] dy^{a} dy^{b}\,.
\end{aligned}\ee
To get the form of the physical metric \(\hat g_{ab} = \Omega^{-2} g_{ab}\) we use \cref{eq:Omega-BS} and define the Beig-Schmidt coordinate \(\rho_\bs \defn 1/\rho\) to obtain 
\be\begin{aligned}
    \hat g_{ab} & \equiv \lb[ 1 + \frac{\sigma}{\rho_\bs} + o(1/\rho_\bs) \rb]^2 d \rho_\bs^{2} + \rho_\bs o(1/\rho_\bs) d \rho_\bs d y^{a} \\
    &\quad + \rho_\bs^{2}\lb[ h_{ab}^{(0)} + \frac{h_{ab}^{(1)}}{\rho_\bs} + o(1/\rho_\bs) \rb] dy^{a} dy^{b}\,,
\end{aligned}\ee
This is the form of the physical metric assumed by Beig and Schmidt \cite{Beig-Schmidt}.

The asymptotic potentials \cref{eq:potentials-defn} are related to the metric coefficients in the above expansion by
\be\label{eq:EK-to-BS}
	\dd E \equiv 4 \sigma \eqsp \dd K_{ab} \equiv -2 ( h_{ab}^{(1)} + 2 \sigma h_{ab}^{(0)})\,.
\ee
From these the asymptotic Weyl tensors can be computed using \cref{eq:EB-potentials}. Note that the parity condition \cref{eq:parity-E} imposed on \(\dd E\) to eliminate the logarithmic translation ambiguity then corresponds to
\be\label{eq:parity-sigma}
    \sigma(\tau,\theta^A) = \sigma(-\tau, -\theta^A)\,.
\ee
From \cref{eq:EK-to-BS} it straightforward to see that our charges for supertranslations \cref{eq:st-charge} matches the expression obtained by Comp\`ere and Dehouck, Eq.~4.88 of \cite{CD}.

For the ``subleading'' magnetic Weyl tensor \(\dd \beta_{ab}\) (defined by \cref{eq:beta-defn} when \(\dd B_{ab} = 0\)) to exist, we need additional regularity conditions on the metric expansion \cref{eq:g-exp}. Thus, to define \(\dd\beta_{ab}\) we assume the ``next order'' expansion
\be\label{eq:g-exp-2}\begin{aligned}
    g_{ab} & \equiv \lb[ 1 + \sigma \rho + o(\rho) \rb]^2 d\rho^2 + \rho\, o(\rho) d \rho d y^{a} \\
    &\quad +\rho^{2}\lb[ h_{ab}^{(0)} + \rho h_{ab}^{(1)} + \rho^2 h_{ab}^{(2)} + o(\rho^2) \rb] dy^{a} dy^{b}\,,
\end{aligned}\ee
where \(h_{ab}^{(2)}\) is a smooth tensor on \(\hyp\). Then, we have (using \(\dd B_{ab} = 0\))
\be \label{eq:h2-beta}
    \dd\beta_{ab} = \dd\varepsilon_{cd(a} \dd D^{c} h^{(2) d}_{b)} - \tfrac{1}{8} \dd\varepsilon_{cd(a} \dd D^{c} \dd E  \dd K^d{}_{b)} - \tfrac{1}{16} \dd\varepsilon_{cd(a}  \dd D_{b)} \dd K^{ce} \dd K^d{}_e\,.
\ee
When the conformal factor is chosen so that \(\dd K_{ab} = 0\), the above expression simplifies considerably. In this case, our Lorentz charge matches the one found by Comp\`ere, Dehouck and Virmani \cite{CDV-Lorentz}. We discuss the case when \(\dd B_{ab}\neq 0\) in \cref{sec:new-beta}.

%%=============================================================
\section{Some useful relations on \(\hyp\)}
\label{sec:useful-H}

In this appendix we collect some relations on the unit-hyperboloid \(\hyp\) which are useful in the main paper.

The Riemann tensor of \(\hyp\) is given by
\be\label{eq:Riem-hyp}
    \dd{\mc R}_{abcd} = \dd h_{ac} \dd h_{bd} - \dd h_{ad} \dd h_{bc}\,.
\ee
Using the above it is easy to derive simple expressions for commuting derivatives on tensor fields on \(\hyp\), see Appendix~A of \cite{Beig-Schmidt}.

%%-------------------------------------------------------------
\subsection{Killing vector fields}
\label{sec:Killing-H}

Let \(\dd X^a\) be a Killing vector field on \(\hyp\), so that \(\dd D_{(a} \dd X_{b)} = 0 \). For any Killing vector field using Eq.~C.3.6 of \cite{Wald-book} and \cref{eq:Riem-hyp} we have
\be\label{eq:DDX}
    \dd D_a \dd D_b \dd X_c  = \dd{\mc R}_{cbad} \dd X^d = \dd h_{ac} \dd X_b - \dd h_{ab} \dd X_c\,.
\ee
Contracting the indices \(a\) and \(b\) we get
\be\label{eq:box-X}
    \dd D^2 \dd X_a + 2 \dd X_a = 0\,.
\ee

Define the ``dual'' vector field \(\dv X^a\) on \(\hyp\) for any Killing vector field \(\dd X^a\) by
\be\label{eq:dual-X-defn}
\dv X^a \defn \tfrac{1}{2} \dd\varepsilon^{abc} \dd D_b \dd X_c\,.
\ee
Then, using \cref{eq:DDX} we have
\be\label{eq:X-dual-X}\begin{aligned}
	\dd D_a \dv X_b = \dd\varepsilon_{abc} \dd X^c \eqsp \dd X^a = - \tfrac{1}{2} \dd\varepsilon^{abc} \dd D_b \dv X_c = - \dv (\dv X)^a \eqsp \dd D_a \dd X_b = - \dd\varepsilon_{abc} \dv X^c\,.
\end{aligned}\ee
In particular \(\dd D_{(a} \dv X_{b)} = 0\) so \(\dv X^a\) is also a Killing vector field on \(\hyp\). In a suitable choice of coordinates on \(\hyp\) this relation maps Lorentz rotations and Lorentz boosts into each other, see Appendix~B of \cite{CDV-Lorentz}.\\

The relationship between the Killing vector fields on \(\hyp\) and Lorentz transformations in the tangent space \(Ti^0\) is as follows. Let \(\dd\Lambda_{ab}\) be a \emph{direction-independent} antisymmetric tensor at \(i^0\) corresponding to an infinitesimal Lorentz transformation in \(Ti^0\). Then the \emph{direction-dependent} vector field defined by\footnote{The relation \cref{eq:X-Lambda} is the ``dual'' of the relation used below Eq.~27 of \cite{AH}.}
\be\label{eq:X-Lambda}
    \dd X^a(\vec\eta) \defn \dd\Lambda^{ab} \dd\eta_b\,,
\ee
is tangent to \(\hyp\). Further, since \(\dd\Lambda_{ab}\) is direction-independent, \(\dd\partial_c \dd\Lambda_{ab} = 0\). Projecting the indices of \(\dd\partial_c \dd\Lambda_{ab} = 0\) tangent and normal to \(\hyp\) in all possible ways it follows that \(\dd X^a\) is a Killing vector field on \(\hyp\) and
\be
    \dd\Lambda_{ab} = - \dd D_a \dd X_b - 2 \dd\eta_{[a} \dd X_{b]}  = \dd\varepsilon_{abc} \dv X^c + \dd\eta_{[a} \dd\varepsilon_{b]cd} \dd D^c \dv X^d\,, 
\ee
where the last equality uses \cref{eq:X-dual-X}. Similarly, it can be shown that if \(\dd X^a\) is the Killing vector field on \(\hyp\) corresponding to \(\dd\Lambda_{ab}\) through \cref{eq:X-Lambda}, then \((- \dv X^a)\) is the Killing vector field on \(\hyp\) corresponding to the ``dual'' Lorentz transformation \(*\dd\Lambda_{ab} \defn \tfrac{1}{2} \dd\varepsilon_{ab}{}^{cd}\dd\Lambda_{cd}\).

%%--------------------------------------------------------
\subsection{Symmetric tensors}
\label{sec:symm-tensors}

Let \(\dd T_{ab}\) be any symmetric tensor on \(\hyp\). Then \(\dd T_{ab}\), its curl and divergence are related by the identity
\be\label{eq:IBP1}
    - 2 \dd T_{ab} \dv X^b + 2 \dd\varepsilon_{cd(a} \dd D^c \dd T^d{}_{b)} \dd X^b - \dd D_c \dd T^{cb} \dd D_a \dv X_b =  \dd D^b \lb( \dd\varepsilon_{abc} \dd T^c{}_d \dd X^d + 2 \dd T^c{}_{[a} \dd D_{b]} \dv X_c \rb)\,,
\ee
where \(\dd X^a\) is any Killing vector on \(\hyp\) and \(\dv X^a\) is the corresponding ``dual'' Killing vector (\cref{eq:dual-X-defn}). This identity can be verified by expanding out the right-hand-side and using \cref{eq:Riem-hyp,eq:dual-X-defn,eq:X-dual-X}. Note that the right-hand-side of \cref{eq:IBP1} corresponds to an exact \(2\)-form on \(\hyp\), and thus vanishes when integrated over any cross-section \(S\) of \(\hyp\). This gives the following useful integral identity on any cross-section \(S\)
\be\label{eq:IBP}
    \int_S \dd\varepsilon_2~ \dd u^a \dd T_{ab} \dv X^b =  \int_S \dd\varepsilon_2~ \dd u^a \lb[ \dd\varepsilon_{cd(a} \dd D^c \dd T^d{}_{b)} \dd X^b - \tfrac{1}{2}\dd D_c \dd T^{cb} \dd D_a \dv X_b \rb]\,.
\ee

In the following lemma we show that any symmetric, curl-free tensor on \(\hyp\) admits a scalar potential. A proof using a choice of coordinates on \(\hyp\) can be found in Appendix~A of \cite{CDV-Lorentz}. Our proof below is adapted from similar arguments for a \(2\)-sphere in Appendix~A.4 of \cite{AK}. 

\begin{lemma}\label{lem:scalar-pot}
Let \(\dd T_{ab}\) be a symmetric tensor on \(\hyp\) with vanishing curl, i.e, \(\dd D_{[c} \dd T_{a]b} = 0\) then there exists a function \(\dd t\) on \(\hyp\) such that
\be
    \dd T_{ab} = \dd D_a \dd D_b \dd t + \dd h_{ab} \dd t\,.
\ee 
\end{lemma}
\begin{proof}
Let \(\dd f \in \mf t\) be a Spi-translation so that\footnote{As shown in \cref{rem:trans-vectors} Spi-translations can also be represented as vectors in the tangent space at \(i^0\).}
\be\label{eq:trans-cond2}
    \dd D_a \dd D_b \dd f + \dd h_{ab} \dd f = 0\,.
\ee

Note that the vector field \(\dd Y^a \defn \dd D^a \dd f\) is a conformal Killing field on \(\hyp\). Any conformal Killing field is completely determined by its \emph{conformal Killing data} specified at some chosen point \(p \in \hyp\) \cite{AMA-isometries}, which in this case is given by
\be\label{eq:conf-Killing-data}
    (\dd Y^a, \dd D_{[a} \dd Y_{b]}, \dd D_a \dd Y^a, \dd D^a \dd D_b \dd Y^b ) \big\vert_p = (\dd D^a \dd f, 0, -3 \dd f, - 3 \dd D^a \dd f) \big\vert_p\,.
\ee
Thus, there is an isomorphism between the vector space of \(\dd f \in \mf t\) and the vector space of the conformal Killing data \(\dd f\vert_p\) and \(\dd D^a \dd f\vert_p\) at any chosen point \(p\).

Since \(\dd T_{ab}\) is symmetric and curl-free, using \cref{eq:trans-cond2} we have \(\dd D_{[c}( \dd T_{a]b} \dd D^b \dd f) = 0\). Thus, \(\dd T_{ab} \dd D^b \dd f\) is a closed \(1\)-form on \(\hyp\) and thus exact,\footnote{This follows from the fact that every \(1\)-loop in \(\hyp\) is contractible to a point and hence the first de Rahm cohomology group of \(\hyp\) is trivial.} that is, there exists a function \(\dd H\) such that
\be\label{eq:T-H}
    \dd T_{ab} \dd D^b \dd f = \dd D_a \dd H\,.
\ee
Thus, \(\dd T_{ab}\) can be viewed as a linear map from the vector space of Spi-translations to functions on \(\hyp\). Since the vector space of Spi-translations is isomorphic to the space of conformal Killing data \cref{eq:conf-Killing-data} specified at any point on \(\hyp\), there exists a function \(\dd t\) and a covector field \(\dd t_a\) on \(\hyp\) such that
\be
    \dd H = \dd t\dd f +  \dd t_a \dd D^a \dd f\,. 
\ee
Inserting this into \cref{eq:T-H} and using \cref{eq:trans-cond2} we get
\be
    \dd T_{ab} \dd D^b \dd f = (\dd D_a \dd t_b + \dd h_{ab} \dd t) \dd D^b \dd f + \dd f (\dd D_a \dd t - \dd t_a)\,.
\ee
Since the conformal Killing data \(\dd f\vert_p\) and \(\dd D^a \dd f\vert_p\) can be freely specified at any point it follows that \(\dd t_a = \dd D_a \dd t\) and 
\be
    \dd T_{ab} = \dd D_a \dd D_b \dd t + \dd h_{ab} \dd t\,.
\ee
\end{proof} 

Note that the potential \(\dd t\) is not uniquely determined,  since one is free to add solutions of \cref{eq:trans-cond2} to \(\dd t\) without affecting the tensor \(\dd T_{ab}\). Further, the potential is \emph{not} locally and covariantly determined by \(\dd T_{ab}\) and finitely many of its derivatives. In particular, even if \(\dd T_{ab}\) is the (direction-dependent) limit to \(i^0\) of some tensor field on spacetime, there may not exist any tensor on spacetime whose limit gives the potential \(\dd t\).

%%-----------------------------------------------------------
\subsection{Closed and exact forms}
\label{sec:Wald-thm}

For some results in the main paper we need to argue that certain \(2\)-forms on \(\hyp\) which are closed are also exact, so that their integral on cross-sections of \(\hyp\) vanishes. In general, not all closed \(2\)-forms on \(\hyp\) are exact since the topology of \(\hyp\) is \(\bb S^2 \times \bb R\) and the second de Rahm cohomology group is nontrivial. However, when the closed \(2\)-forms considered are local and covariant functionals of suitable fields (as described below) then they can be shown to be exact by a general theorem of Wald \cite{W-closed}.

In the theorem stated below, the differential forms \(\mu[\phi,\psi]\) under consideration will be functionals of two types of fields. The ``dynamical fields'', denoted by \(\phi\), are arbitrary cross-sections of some vector bundle, and we require that \(d\mu = 0\) for every cross-section \(\phi\). The form \(\mu\) also can depend on some ``background fields'', denoted by \(\psi\). The ``background fields'' \(\psi\) need not have a linear structure and are allowed to satisfy (possibly nonlinear) differential equations. Now we can state the theorem from \cite{W-closed}.

\begin{thm}[\cite{W-closed}]\label{thm:Wald}
Let \(\mu[\phi,\psi]\) be a \(p\)-form on a \(d\)-dimensional manifold \(M\) with \(p < d\), which is a local and covariant functional of a collection of two sets of fields \((\phi,\psi)\) (as described above) and finitely many of their derivatives on \(M\). Then, if for \emph{any} ``background fields'' \(\psi\) 
\begin{enumerate}
    \item \(d\mu[\phi,\psi] = 0\) for \emph{all} cross-sections of the vector bundle of ``dynamical fields'' \(\phi\) and
    \item \(\mu[\phi, \psi] = 0\) for the zero cross-section \(\phi = 0\)
\end{enumerate}
then there exists a \((p-1)\)-form \(\nu[\phi,\psi]\) which is a local and covariant functional of \((\phi,\psi)\) and finitely many of their derivatives such that \(\mu[\phi,\psi] = d \nu[\phi,\psi]\). That is the closed \(p\)-form \(\mu\) is also exact. 
\end{thm}
Note that it is essential for this theorem that the ``dynamical fields'' have a linear structure as the cross-sections of some vector bundle and further, the \(p\)-form \(\mu\) must be closed for all possible cross-sections of this vector bundle, i.e., one must be able to freely specify the ``dynamical fields'' and all of their derivatives at any point of \(M\). In contrast, the ``background fields'' \(\psi\), need not have a linear structure and are allowed to satisfy differential equations, and in fact the set of ``background fields'' can also be empty. Further, the proof in \cite{W-closed} also provides a constructive procedure for finding the \((p-1)\)-form \(\nu\) though we will not need to use this construction.

For our applications of this theorem we will be concerned with closed \(2\)-forms on \(\hyp\). Using the volume element \(\dd\varepsilon_{abc}\) on \(\hyp\), we will write this \(2\)-form in terms of a covector \(\dd\mu_a\) such that \(\dd D^a \dd\mu_a = 0\). Then, from \cref{thm:Wald} we conclude that this \(2\)-form is exact and thus
\be\label{eq:Wald-hyp}
    \dd D^a \dd\mu_a[\phi,\psi] = 0 \implies \int_S \dd\varepsilon_2~ \dd u^a \dd \mu_a [\phi, \psi] = 0\,,
\ee
for any cross-section \(S\) of \(\hyp\) with \(\dd\varepsilon_2\) and \(\dd u^a\) being the area element and normal to \(S\). The choice of the ``dynamical fields'' \(\phi\) depends on the particular case. Since the fields  \(\dd E\), \(\dd K_{ab}\) and \(\dd\beta_{ab}\) satisfy differential equations of motion (\cref{eq:box-E,eq:K-eom,eq:div-beta}) they cannot be used as the ``dynamical fields''. Similarly, the Lorentz vector fields \(\dd X^a\) form a \(6\)-dimensional vector space and cannot be arbitrary sections of some vector bundle and also cannot be used as the ``dynamical fields''. Thus, these fields, along with the metric and volume form on \(\hyp\), will always be in the collection of ``background fields'' \(\psi\).

 However, the supertranslation symmetries \(\dd f\), the freedom in the conformal factor \(\dd\alpha\) (\cref{rem:conf}) and the scalar potential \(\dd k\) for \(\dd K_{ab}\) (when \(\dd B_{ab} = 0\)) are free functions on \(\hyp\) and will be used as ``dynamical fields'' in our applications of this theorem.

%%===========================================================
\section{Conformal transformation of the Lorentz charges}
\label{sec:conf-lorentz-charge}

In \cref{sec:conf-charges} we argued that under conformal transformations the Lorentz charge shifts by the charge of a supertranslation (\cref{eq:lor-charge-trans-summ}). In this appendix we collect the explicit computation of this transformation.

Using \cref{eq:conf-K,eq:conf-beta}, and that \(\dd E\) is conformally-invariant, we have the following transformation for the tensor \(\dd W_{ab}\) defined in \cref{eq:W-defn} under changes of the conformal factor
\be\label{eq:conf-W}
    \dd W_{ab} \mapsto \dd W_{ab} + \tfrac{1}{4} \dd\varepsilon_{cd(a} \dd D^c \lb[ \dd D_{b)} \dd E \dd D^d \dd\alpha + \dd D^d \dd E \dd D_{b)} \dd\alpha \rb] + \tfrac{1}{8} \dd\varepsilon_{abc} \dd D^c \dd E (\dd D^2 \dd\alpha + 3 \dd\alpha)\,. 
\ee
Thus, we have (note that the Lorentz vector does not transform under changes of the conformal factor \cref{rem:conf-symm})
\be\label{eq:W-trans1}
    \dd W_{ab} \dv X^b \mapsto \dd W_{ab} \dv X^b + \tfrac{1}{4} \dd\varepsilon_{cd(a} \dd D^c \dd T^d{}_{b)} \dv X^b + \tfrac{1}{8} (\dd D^2 \dd\alpha + 3 \dd\alpha) \dd D^b \dd E \dd D_a \dd X_b\,,
\ee
where we have defined the shorthand \(\dd T_{ab} \defn \dd D_a \dd E \dd D_b \dd\alpha + \dd D_b \dd E \dd D_a \dd\alpha\) and used the last identity in \cref{eq:X-dual-X}. Now using the identity \cref{eq:IBP} (with \(\dd X^a\) replaced by \(\dv X^a\)) we have
\be
    \int_S \dd\varepsilon_2~ \dd u^a \dd\varepsilon_{cd(a} \dd D^c \dd T^d{}_{b)} \dv X^b = - \int_S \dd\varepsilon_2~ \dd u^a \lb[ \tfrac{1}{2} \dd D_c \dd T^{cb} \dd D_a \dd X_b + \dd T_{ab} \dd X^b \rb]\,. 
\ee
A straightforward but tedious computation using the definition of \(\dd T_{ab}\), \cref{eq:box-E,eq:Riem-hyp,eq:box-X} gives
\be\begin{aligned}
    \int_S \dd\varepsilon_2~ \dd u^a \dd\varepsilon_{cd(a} \dd D^c \dd T^d{}_{b)} \dv X^b = \int_S \dd\varepsilon_2~ \dd u^a  \big[ & -\tfrac{1}{2} (\dd D^2 \dd\alpha + 3 \dd\alpha) (2 \dd E \dd X_a + \dd D^b \dd E \dd D_a \dd X_b ) \\
    & + \lb( \dd E \dd D_a \lie_{\dd X} \dd\alpha - \dd D_a \dd E \lie_{\dd X} \dd\alpha \rb) \big]\,, 
\end{aligned}\ee
where we have dropped terms that integrate to zero on \(S\). Using the above in \cref{eq:W-trans1} we get
\be\label{eq:W-trans2}\begin{aligned}
    \int_S \dd\varepsilon_2~ \dd u^a \dd W_{ab} \dv X^b \mapsto \int_S \dd\varepsilon_2~ \dd u^a \dd W_{ab} \dv X^b + \tfrac{1}{4} \int_S \dd\varepsilon_2~ \dd u^a \big[ & \lb( \dd E \dd D_a \lie_{\dd X} \dd\alpha - \dd D_a \dd E \lie_{\dd X} \dd\alpha \rb) \\
    & -  (\dd D^2 \dd\alpha + 3 \dd\alpha) \dd E \dd X_a \big]\,.
\end{aligned}\ee
Further, from \cref{eq:conf-K} we also have
\be
    - \tfrac{1}{8} \dd K \dd E \dd X_a \mapsto - \tfrac{1}{8} \dd K \dd E \dd X_a + \tfrac{1}{4} (\dd D^2 \dd\alpha + 3 \dd\alpha) \dd E \dd X_a\,.
\ee
Thus,
\be\label{eq:lor-charge-trans1}\begin{aligned}
    \int_S \dd\varepsilon_2~ \dd u^a \lb[\dd W_{ab} \dv X^b  - \tfrac{1}{8} \dd K \dd E \dd X_a \rb] &\mapsto \int_S \dd\varepsilon_2~ \dd u^a \lb[\dd W_{ab} \dv X^b  - \tfrac{1}{8} \dd K \dd E \dd X_a \rb] \\
    &\qquad + \tfrac{1}{4} \int_S \dd\varepsilon_2~ \dd u^a \lb( \dd E \dd D_a \lie_{\dd X} \dd\alpha - \dd D_a \dd E \lie_{\dd X} \dd\alpha \rb)\,.
\end{aligned}\ee
The Lorentz charge \cref{eq:lorentz-charge} then transforms as 
\be\label{eq:lor-charge-trans}\begin{aligned}
    \mc{Q}[\dd X^a; S] & \mapsto \mc Q[\dd X^a; S] - \frac{1}{16\pi}\int_S \dd\varepsilon_{2}~ \dd u^{a} \tfrac{1}{2} \lb( \dd E \dd D_a \lie_{\dd X} \dd\alpha - \dd D_a \dd E \lie_{\dd X} \dd\alpha \rb)\,. 
\end{aligned}\ee
Comparing to \cref{eq:st-charge}, we recognize the last integral above as the charge of the supertranslation \((-\half \Lie_{\dd X}\dd\alpha)\). Thus, the Lorentz charge shifts by the charge of a supertranslation under changes of the conformal factor as argued in \cref{sec:conf-charges}.

%%===========================================================
\section{Ambiguities in the Spi-charges}
\label{sec:amb}

In this section we analyze the ambiguities in our procedure to define the Spi charges. We show our Spi charges are unambiguously defined by the choice of the symplectic current for general relativity in \cref{eq:sympcurrent-H,eq:sympcurrent-H-simplified}.

Recall that our charges on a cross-section \(S\) of \(\hyp\) are defined by
\be\label{eq:WZ-charge2}
    \delta \mc Q[(\dd f, \dd X^a) ; S] \defn \int_S \dd\varepsilon_2 \dd u^a \dd Q_a(\delta g; (\dd f, \dd X^a)) + \int_S \dd X \cdot \dd\Theta(\delta g)\,, 
\ee
with \(\mc Q = 0 \) on Minkowski spacetime as the reference solution. The covector \(\dd Q_a\) is a local and covariant functional of its arguments and linear in the metric perturbations and the asymptotic symmetry satisfying \cref{eq:omega-Q}. While the \(3\)-form \(\dd\Theta\) is a symplectic potential for \(\pb{\dd\omega}\) satisfying \cref{eq:Theta-symppot}.

Given a fixed choice of the symplectic current, from \cref{eq:omega-Q,eq:Theta-symppot} the ambiguities in the choice of \(\dd Q_a\) and the \(\dd\Theta\) are given by
\be\begin{aligned}
    \dd Q_a(g;\delta g; (\dd f, \dd X)) &\mapsto \dd Q_a(g;\delta g; (\dd f, \dd X)) + \dd \mu_a(g;\delta g; (\dd f, \dd X))\,, \\
    \dd\Theta(\delta g) &\mapsto \dd\Theta(\delta g) + \dd\varepsilon_3 \delta \dd \Xi(g)\,,
\end{aligned}\ee
where the covector \(\dd \mu_a (g;\delta g; (\dd f, \dd X)) \) is a local and covariant functional of its arguments and linear in the metric perturbations and the asymptotic symmetry, and further satisfies
\be\label{eq:mu-cond}
    \dd D^a \dd \mu_a (g;\delta g; (\dd f, \dd X)) = 0\,,
\ee
for all background spacetimes and perturbations (satisfying the background and linearized equations of motion respectively) and all asymptotic symmetries. While the function \(\dd\Xi\) is any local and covariant function of the background spacetime fields on \(\hyp\).

Under these ambiguities the definition of \(\delta\mc Q\) (\cref{eq:WZ-charge2}) changes by
\be\label{eq:WZ-amb}
    \delta \mc Q[(\dd f, \dd X^a) ; S] \mapsto \delta \mc Q[(\dd f, \dd X^a) ; S] + \int_S \dd\varepsilon_2 \dd u^a \dd \mu_a(g;\delta g; (\dd f, \dd X)) - \delta \int_S \dd\varepsilon_2 \dd u^a \dd X_a \dd\Xi(g)\,.
\ee
Since the integrated charge \(\mc Q\) is fixed by the requirement that it vanish on Minkowski spacetime (where \(\dd E = \dd \beta_{ab} = 0\)), we only need to analyze the ambiguities in \(\delta\mc Q\).  

We now argue that the last two integrals above must vanish under the following assumptions
\begin{enumerate}
    \item \(\dd\mu_a\) and \(\dd\Xi\) are local and covariant functionals of their arguments as mentioned above with \(\dd\mu_a\) satisfying \cref{eq:mu-cond}.
    \item The Lorentz charge \(\mc Q[(\dd f=0, \dd X^a) ; S]\) must match the Ashtekar-Hansen expression when the conformal factor is chosen such that \(\dd K_{ab} = 0\).
    \item The total charge \(\mc Q[(\dd f, \dd X^a) ; S]\) of any Spi symmetry is conformally-invariant.
\end{enumerate}

Consider first the \(\dd\mu_a\)-ambiguity and the case of a pure supertranslation \((\dd f, \dd X^a = 0)\). Since the ambiguity \(\dd \mu_a\) is linear in \(\dd f\) we have \(\dd \mu_a(g; \delta g; \dd f = 0) = 0 \). Further since \(\dd\mu_a\) is divergence-free (\cref{eq:mu-cond}), we can use \cref{thm:Wald} in the form \cref{eq:Wald-hyp} with \(\dd f\) as ``dynamical field'' to conclude that the second integral on the right-hand-side of \cref{eq:WZ-amb} vanishes on any cross-section \(S\) for a supertranslation.

Next consider the \(\dd\mu_a\)-ambiguity with a Lorentz transformation \( (\dd f =0 , \dd X^a)\). Since the Lorentz vector fields \(\dd X^a\) form a \(6\)-dimensional vector space and are not allowed to be arbitrary cross-sections of a vector bundle, we cannot use \(\dd X^a\) as the ``dynamical'' fields in \cref{thm:Wald}. So instead, we proceed another in another way. Consider the scalar potential \(\dd k\) for the tensor \(\dd K_{ab}\) (\cref{eq:K-potential}). Since \(\dd k\) is a completely free function on \(\hyp\) it is allowed to be an arbitrary cross-section of a vector bundle on \(\hyp\). Further, whenever \(\dd k=0\) we have \(\dd K_{ab} = 0\) and by our assumption the Lorentz charge must the one found by Ashtekar and Hansen. Thus, the ambiguity \(\dd \mu_a = 0\) whenever \(\dd k = 0\) for all background spacetimes and all Lorentz vector fields \(\dd X^a\). Now using \(\dd k\) as the ``dynamical field'', from \cref{thm:Wald}  in the form \cref{eq:Wald-hyp}, we conclude again that the second integral on the right-hand-side of \cref{eq:WZ-amb} vanishes on any cross-section \(S\) for a Lorentz symmetry. Thus, the \(\dd\mu_a\)-ambiguity does not affect \(\delta \mc Q\).

Finally, consider the \(\dd\Xi\)-ambiguity in the choice of \(\dd\Theta\). In \cref{sec:conf-charges} we showed that the total charge \(\mc Q\) for any Spi-symmetry \((\dd f, \dd X^a)\) is invariant under conformal transformations with our choice of \(\dd\Theta\) (\cref{eq:Theta-choice}) which implies that the charge of a ``pure Lorentz'' symmetry must shift by a charge of a supertranslation under changes of the conformal factor (see \cref{eq:lor-charge-trans-summ}). It follows that for the redefined Lorentz charge to transform correctly the integral contributed by \(\dd\Xi\) in \cref{eq:WZ-amb} must be conformally-invariant. Further, for the redefined Lorentz charge to match the one found by Ashtekar and Hansen the integral contributed by \(\dd\Xi\) in \cref{eq:WZ-amb} must vanish whenever \(\dd K_{ab} = 0\). Since \(\dd K_{ab}\) can be chosen to vanish by a choice of conformal factor (see \cref{rem:GR-gauge-choice}) this implies the \(\dd\Xi\)-ambiguity does not affect \(\delta\mc Q\).

In summary, our charges are unambiguously determined by the pullback of the symplectic current \cref{eq:sympcurrent-H-simplified}.

Here we remark that the symplectic current \(3\)-form itself is \emph{not} uniquely determined by the Lagrangian of the theory but is ambiguous up to
\be
    \omega(g; \delta_1 g, \delta_2 g) \mapsto \omega(g; \delta_1 g, \delta_2 g) + d \lb[ \delta_1 \nu(g;\delta_2 g) - \delta_2 \nu(g;\delta_1 g) \rb]\,,
\ee
where \(\nu(g;\delta g)\) is a local and covariant \(2\)-form and is linear in the perturbation \(\delta g\). We have not analyzed the effect of this ambiguity on our charges.

%%===========================================================
\section{Lorentz charges with \(\dd B_{ab} \neq 0\)}
\label{sec:new-beta}

In \cref{sec:lorentz-charge} to define the Lorentz charges at \(i^0\) we imposed the condition \(\dd B_{ab} =0 \) to gain access to the ``subleading'' magnetic part \(\dd\beta_{ab}\) of the asymptotic Weyl tensor (see \cref{eq:beta-defn}). In this section we show how we can define a ``subleading'' magnetic Weyl tensor and the Lorentz charges even when \(\dd B_{ab} \neq 0\).

If \(\dd B_{ab}\) does not vanish, then the ``subleading'' piece as defined by \cref{eq:beta-defn} does not exist in the limit. However, consider the derivative of the magnetic part of the Weyl tensor along \(\eta^a\):
\be
    \lim_{\to i^0} \Omh \eta^e \nabla_e (\Omh * C_{acbd} \eta^c \eta^d) = \dd\eta^e \dd\partial_e \dd B_{ab} = 0\,.
\ee
Since the limit of the above quantity vanishes we can now demand that its ``next order'' part exist, that is, 
\be\label{eq:H-defn}
    \dd H_{ab}(\vec\eta) \defn \lim_{\to i^0} \eta^e \nabla_e (\Omh * C_{acbd} \eta^c \eta^d) \quad\text{is } C^{>-1}\,.
\ee

The tensor field \(\dd H_{ab}(\vec\eta)\) is not tangential to \(\hyp\). We can compute
\be\label{eq:H-eta}\begin{aligned}
    \dd H_{ab}(\vec\eta) \dd\eta^b & = \lim_{\to i^0} \eta^b \eta^e \nabla_e (\Omh * C_{acbd} \eta^c \eta^d)
    = - \lim_{\to i^0}  \eta^e \nabla_e \eta^b (\Omh * C_{acbd} \eta^c \eta^d) \\
    & = \tfrac{1}{4} \dd B_{ab} \dd D^b \dd E\,,
\end{aligned}\ee
where in the first line we have used the fact that \(*C_{abcd}\) is antisymmetric in the last two indices and to get the second line we replaced the derivative of \(\eta^a\) using the Einstein equation \cref{eq:EE}, and used \cref{eq:EB-defn,eq:h-eta-S}. Note that \(\dd H_{ab}(\vec\eta)\dd\eta^a \dd\eta^b = 0\), and thus the only remaining part of \(\dd H_{ab}\) is its projection to \(\hyp\) on both indices. We use this projection to define the ``subleading'' magnetic part of the Weyl tensor, that is, instead of \cref{eq:beta-defn} we now use
\be\label{eq:new-beta-defn}
    \dd \beta_{ab} \defn \dd h_a{}^c \dd h_b{}^d \dd H_{cd}(\vec\eta)\,.
\ee
As before \(\dd\beta_{ab}\) is a symmetric and traceless tensor field on \(\hyp\). Note that when \(\dd B_{ab}=0\), this new definition is completely equivalent to the previous one in \cref{eq:beta-defn} (see also \cite{AH}).

The generalization of the equation of motion \cref{eq:div-beta} is rather tedious to obtain. We want to compute
\be\begin{aligned}
    \dd\partial^b \dd H_{ab} & = \lim_{\to i^0} \Omh \nabla^b \lb[ \eta^e \nabla_e (\Omh * C_{acbd} \eta^c \eta^d) \rb] \\
    & = \lim_{\to i^0} \lb[ (\nabla^b \eta^e) \Omh \nabla_e (\Omh * C_{acbd} \eta^c \eta^d) + \Omh \eta^e \nabla^b \nabla_e (\Omh * C_{acbd} \eta^c \eta^d) \rb]\,.
\end{aligned}\ee
In the first term we substitute the derivative of \(\eta^a\) using \cref{eq:EE} and then evaluate the limit of the expression using \cref{eq:E-B-decomp,eq:potentials-defn,eq:h-eta-S,eq:H-eta}. For the second term on the right-hand-side, we first commute the derivatives and introduce terms involving the the Riemann tensor of the unphysical spacetime. The term with the derivatives \(\nabla^b\) and \(\nabla_e\) interchanged vanishes in the limit while the Riemann tensor terms can be computed by decomposing the Riemann tensor in terms of the Weyl tensor \(C_{abcd}\) and \(S_{ab}\) (\cref{eq:S-defn}). Then we can evaluate the limit using \cref{eq:E-B-decomp,eq:potentials-defn,eq:h-eta-S}. The final limit gives the equation
\be
    \dd\partial^b \dd H_{ab} = - \tfrac{1}{4} \dd\partial_c \dd B_{ab} \dd K^{bc} - \tfrac{1}{4} \dd\partial^b \dd B_{ab} \dd E + \tfrac{5}{4} \dd B_{ab} \dd D^b \dd E + \tfrac{1}{4} \dd\varepsilon_{cda} \dd E^c{}_b \dd K^{db} - \tfrac{1}{4} \dd\eta_a \dd B_{bc} \dd K^{bc} - \dd\eta_a \dd B_{bc} \dd E^{bc}\,.
\ee
Using \cref{eq:H-eta} and the equation of motion \cref{eq:EB-curl} it can be verified that the contraction of the above equation with \(\dd\eta^a\) is trivial. Projecting the index \(a\) on to \(\hyp\) we then get the equation of motion for \(\dd\beta_{ab}\) as
\be\label{eq:new-div-beta}
    \dd D^b \dd \beta_{ab} = \tfrac{1}{4} \dd\varepsilon_{cda} \dd E^c{}_b \dd K^{bd} + \tfrac{5}{4} \dd B_{ab} \dd D^b \dd E - \tfrac{1}{4} \dd D_a \dd B_{bc} \dd K^{bc}\,,
\ee
which reduces to \cref{eq:div-beta} when \(\dd B_{ab} = 0\).

To define the Lorentz charge we now construct the generalization of the tensor \(\dd W_{ab}\) (\cref{eq:W-defn}). Note that the only essential properties of \(\dd W_{ab}\) used to obtain \cref{eq:symp-lorentz-W} are that \(\dd W_{[ab]} = - \tfrac{1}{16} \dd{\varepsilon}_{abc}\dd{K}\dd D^{c}\dd{E}\) and \(\dd D^a \dd W_{ab} = 0\) using the equation of motion for \(\dd\beta_{ab}\). We will further require that \(\dd W_{ab}\) is also traceless.

To find such a \(\dd W_{ab}\), first note that the last term in \cref{eq:new-div-beta} can be written as the divergence of a symmetric tensor using \cref{eq:EB-div,eq:EB-potentials,eq:K-eom}
\be\label{eq:D-B-K}
    -\tfrac{1}{4} \dd D_b \dd B_{ac} \dd K^{ac} = -\tfrac{1}{16} \dd D^a \lb[ - 2 \dd B_{ab} \dd K +  2\dd h_{ab} \dd B_{cd} \dd K^{cd} - \dd\varepsilon_{cd(a} \dd K^c{}_{b)} \dd D^d \dd K - \dd\varepsilon_{cd(a} \dd D_{b)} \dd K^{ce} \dd K^d{}_e \rb]\,.
\ee
Note that the tensor in the square brackets is \emph{not} traceless. However, we can add to it the following symmetric tensor 
\be
    -\tfrac{5}{8} \lb[ 2 \dd B_{c(a} \dd K_{b)}^c - \dd h_{ab} \dd B_{cd} \dd K^{cd} - \dd B_{ab} \dd K \rb]\,,
\ee
which has vanishing divergence and thus does not affect the left-hand-side. With this we define 
\be\label{eq:new-W-defn}\begin{aligned}
    \dd W_{ab} & \defn \dd{\beta}_{ab} + \tfrac{1}{8} \dd{\varepsilon}_{cd(a} \dd D^{c} \dd{E} \dd K^d{}_{b)} - \tfrac{1}{16} \dd{\varepsilon}_{abc}\dd{K}\dd D^{c}\dd{E}  \\
    &\qquad - \tfrac{3}{2} \dd B_{ab} \dd E + \tfrac{5}{4} \dd B_{c(a} \dd K_{b)}^c - \tfrac{1}{2} \dd h_{ab} \dd B_{cd} \dd K^{cd} - \tfrac{3}{4} \dd B_{ab} \dd K \\
    &\qquad - \tfrac{1}{16} \dd\varepsilon_{cd(a} \dd D_{b)} \dd K^{ce} \dd K^d{}_e - \tfrac{1}{16} \dd\varepsilon_{cd(a} \dd K^c{}_{b)} \dd D^d \dd K\,,
\end{aligned}\ee
which satisfies
\be
    \dd W_{[ab]} = - \tfrac{1}{16} \dd{\varepsilon}_{abc}\dd{K}\dd D^{c}\dd{E} \eqsp \dd D^a \dd W_{ab} = 0 \eqsp \dd h^{ab} \dd W_{ab} = 0\,.
\ee
Then the Lorentz charge formula takes the same form as in \cref{eq:lorentz-charge} with \(\dd W_{ab}\) now defined as in \cref{eq:new-W-defn}. The flux of this charge is still given by the expression \cref{eq:lorentz-flux}.

Note that when \(\dd B_{ab} = 0\), the second line in \cref{eq:new-W-defn} vanishes, but the terms in the third line are nonvanishing in general; denote these terms by a symmetric tensor \(\dd T_{ab}\). It follows from \cref{eq:D-B-K} that \(\dd T_{ab}\) is divergence-free when \(\dd B_{ab} = 0\). Thus \(\dd D^a (\dd T_{ab} \dv X^b) = 0\) and \(\dd T_{ab} \dv X^b = 0\) when the scalar potential \(\dd k\) for \(\dd K_{ab}\) (\cref{eq:K-potential}) vanishes. Using the scalar potential \(\dd k\) as the ``dynamical field'' in \cref{thm:Wald} it follows from \cref{eq:Wald-hyp} that these terms do not contribute to the Lorentz charge expression. Thus, when \(\dd B_{ab} = 0\) the Lorentz charge defined using \cref{eq:new-W-defn} coincides with the one defined previously in \cref{sec:lorentz-charge}.

Under conformal transformations we can show that
\be
    \dd\beta_{ab} \mapsto \dd\beta_{ab} - \dd\varepsilon_{cd(a} \dd E^c{}_{b)} \dd D^d \dd\alpha - \tfrac{3}{2} \dd B_{ab} \dd\alpha + \tfrac{1}{2} \dd D_c \dd B_{ab}  \dd D^c \dd\alpha\,,
\ee
and that \cref{eq:new-div-beta} is invariant. The explicit computation of the transformation of the Lorentz charge presented in \cref{sec:conf-lorentz-charge} now becomes much more complicated. However, the general argument presented in \cref{sec:conf-charges} still holds. Thus, even without the assumption \(\dd B_{ab} = 0\) we have a satisfactory definition of Lorentz charges at spatial infinity.\\

The Lorentz charges for \(\dd B_{ab} \neq 0\) case were also derived by Comp\`ere and Dehouck \cite{CD} (with \(\dd K = 0\)) using an asymptotic expansion in Beig-Schmidt coordinates which in the unphysical spacetime coordinates used in \cref{sec:coord} reads
\be\label{eq:g-exp-CD}\begin{aligned}
    g_{ab} & \equiv \lb[ 1 + \sigma \rho + o(\rho) \rb]^2 d\rho^2 + \rho\, o(\rho) d \rho d y^{a} \\
    &\quad +\rho^{2}\lb[ h_{ab}^{(0)} + \rho h_{ab}^{(1)} - \rho^2 \ln\rho~ i_{ab} + \rho^2 h_{ab}^{(2)} + o(\rho^2) \rb] dy^{a} dy^{b}\,.
\end{aligned}\ee
For \(\dd\beta_{ab}\), as defined by \cref{eq:new-beta-defn}, to exist we set the logarithmic term \(i_{ab} = 0\). With this condition the \(\dd\beta_{ab}\) is related to the curl of the metric coefficient \(h_{ab}^{(2)}\) with additional terms whose form is rather complicated (as compared to \cref{eq:h2-beta} when \(\dd B_{ab} = 0\)). Note that with \(\dd K = 0\), our \(\dd W_{ab}\) is a symmetric, divergence-free and traceless tensor and thus we expect that our charge expression in this case matches with the one derived in \cite{CD} in terms of \(h_{ab}^{(2)}\), but we have not shown this explicitly. 

When the logarithmic term \(i_{ab}\) does not vanish, our definition \cref{eq:new-beta-defn} cannot be used for the ``subleading'' magnetic part of the Weyl tensor. We have not explored this case in detail but we expect the following strategy to be useful. We can assume that
\be
    \Omh *C_{acbd}\eta^c \eta^d = B_{ab} + \Omh \ln\Omh b_{ab} + \Omh \beta_{ab} + o(\Omh)\,,
\ee
where each of the tensors \(B_{ab}\), \(b_{ab}\) and \(\beta_{ab}\) are symmetric and orthogonal to \(\eta^a\) and admit a \(C^{>-1}\) limit to \(i^0\). Using such an expansion in the Hodge dual of \cref{eq:curl-weyl} we can derive the equations of motion for the limits of \(B_{ab}\), \(b_{ab}\) and \(\beta_{ab}\). Since the expression for the symplectic current \cref{symp-lorentz} is unchanged, we can use these equations of motion to define an analogue of the tensor \(\dd W_{ab}\) and the Lorentz charges. From the point of view of matching these charges to those on null infinity, we expect that the spacetimes with such a logarithmic behaviour at spatial infinity would correspond to the polyhomogenous spacetimes at null infinity defined in \cite{CMS-poly}.

%%THE END

%\newpage

\bibliographystyle{JHEP}
\bibliography{spi-charges}      
\end{document}